\documentclass[10pt,journal,compsoc]{IEEEtran}
\pdfoutput=1
\ifCLASSOPTIONcompsoc
  \usepackage[nocompress]{cite}
\else
  \usepackage{cite}
\fi
\ifCLASSINFOpdf
   \usepackage[pdftex]{graphicx}
\else
  \usepackage[dvips]{graphicx}
\fi
\usepackage{amsmath, amssymb, amsthm}
\usepackage{algorithm}
\usepackage[noend]{algpseudocode}

\usepackage{array}

\ifCLASSOPTIONcompsoc
  \usepackage[caption=false,font=footnotesize,labelfont=sf,textfont=sf]{subfig}
\else
  \usepackage[caption=false,font=footnotesize]{subfig}
\fi
\usepackage{url}

\usepackage{booktabs}
\usepackage{multirow}
\usepackage{hyperref}
\hypersetup{
    colorlinks=true,
    linkcolor=black,
    urlcolor=black,
    citecolor=black
}

\newcommand{\noun}[1]{\textsc{#1}}

\newtheorem{definition}{Definition}
\newtheorem{example}{Example}
\newtheorem{theorem}{Theorem}
\begin{document}
%
\title{Distributed Subgraph Enumeration via Backtracking-based Framework}
%
%
%
%

\author{Zhaokang~Wang, Weiwei~Hu, Chunfeng~Yuan, Rong~Gu*, Yihua~Huang*%
  \IEEEcompsocitemizethanks{\IEEEcompsocthanksitem All authors are with the State Key Laboratory for Novel Software Technology, Department of Computer Science and Technology, Nanjing University, China.%
    \IEEEcompsocthanksitem E-mail: \{wangzhaokang, weiweihu\}@smail.nju.edu.cn, \{cfyuan, gurong, yhuang\}@nju.edu.cn.%
    \IEEEcompsocthanksitem Corresponding authors are Yihua Huang and Rong Gu with equal contribution.}%
  \thanks{Manuscript received XXX; revised XXX}}

%
%

\markboth{Technical Report}{Technical Report}
\IEEEtitleabstractindextext{%
  \begin{abstract}
    Finding or monitoring subgraph instances that are isomorphic to a given pattern graph in a data graph is a fundamental query operation in many graph analytic applications, such as network motif mining and fraud detection.
    The state-of-the-art distributed methods are inefficient in communication.
    They have to shuffle partial matching results during the distributed multiway join.
    The partial matching results may be much larger than the data graph itself.
    To overcome the drawback, we develop the Batch-BENU framework (B-BENU) for distributed subgraph enumeration. 
    B-BENU executes a group of local search tasks in parallel.
    Each task enumerates subgraphs around a vertex in the data graph, guided by a backtracking-based execution plan.
    B-BENU does not shuffle any partial matching result.
    Instead, it stores the data graph in a distributed database.
    Each task queries adjacency sets of the data graph on demand.
    To support dynamic data graphs, we propose the concept of incremental pattern graphs and turn continuous subgraph enumeration into enumerating incremental pattern graphs at each time step.
    We develop the Streaming-BENU framework (S-BENU) to enumerate their matches efficiently.
    We implement B-BENU and S-BENU with the local database cache and the task splitting techniques.
    The extensive experiments show that B-BENU and S-BENU can scale to big data graphs and complex pattern graphs.
    They outperform the state-of-the-art methods by up to one and two orders of magnitude, respectively.
 
  \end{abstract}
  
  
  }

\maketitle

\IEEEdisplaynontitleabstractindextext

%
\IEEEpeerreviewmaketitle

\IEEEraisesectionheading{\section{Introduction}\label{sec:introduction}}

\IEEEPARstart{G}{iven} a big data graph $G$ and a small pattern graph $P$, \emph{subgraph enumeration} is to find all the subgraph instances of $G$ that are isomorphic to $P$.
The subgraph instances are the matching results of $P$ in $G$.
Subgraph enumeration is a fundamental query operation in many graph analytic applications, including network motif mining \cite{Milo2002NetworkMS}, graphlet-based network comparison \cite{DBLP:journals/bioinformatics/Przulj07}, network evolution analysis \cite{DBLP:conf/wsdm/KairamWL12}, and social network recommendation\cite{DBLP:journals/pvldb/FanWWX15}.

When the data graph is dynamic, the subgraph enumeration problem becomes the \emph{continuous subgraph enumeration} problem.
The edge set of a dynamic data graph evolves over time.
The matching results of a pattern graph also change consequently.
The continuous subgraph enumeration focus on monitoring the changes in the matching results as the data graph evolves.
Detecting appearing subgraph instances of suspicious pattern graphs timely is essential in real-world applications like fraud detection \cite{neo4j-fraud-detection} \cite{GraphS} and cybersecurity  \cite{cyber-security}.

\subsection{Motivation}

Enumerating instances of a pattern graph in a big data graph is challenging due to two difficulties.
First, the core operation of subgraph enumeration is subgraph isomorphism.
It is an NP-complete problem and has high computational complexity.
Second, the sizes of (partial) matching results can be much larger than the data graph itself  \cite{SEED} \cite{CBF}.
Table \ref{tab:statistics_of_typical_pattern_graphs} shows the numbers of matches of some typical pattern graphs in real-world data graphs.
The numbers of matching results can be 10 to 100 times larger than the numbers of edges in data graphs.
Just scanning matching results takes considerable computational costs.

Some serial in-memory subgraph enumeration algorithms like \cite{TurboISO} \cite{CFLMatch} and out-of-core algorithm \cite{DUALSIM} are proposed, but the computing
power of a single machine limits their performance.
The emerging need to process big data graphs inspires researchers to design efficient distributed subgraph enumeration methods.
Based on whether a distributed algorithm shuffles intermediate results, we divide the existing distributed algorithms into two groups: DFS-style and BFS-style.

The DFS-style algorithms do not shuffle intermediate results.
Instead, they shuffle the data graph.
QFrag \cite{QFrag} broadcasts the data graph to each machine and enumerates subgraphs in memory on each machine concurrently.
However, it cannot scale to data graphs bigger than the memory capacity.
Afrati et al. \cite{UllmanOneRound} use the one-round multiway join to enumerate subgraphs with MapReduce.
However, it cannot scale to complex pattern graphs due to large replication of edges, empirically performing worse than the BFS-style algorithm \cite{TwinTwigJoin} \cite{TwinTwigJoinJournalVersion}.

\begin{table}[!t]
    \caption{Numbers of Matches of Typical Pattern Graphs in Real-world Data Graphs \label{tab:statistics_of_typical_pattern_graphs}}
    \centering
    \begin{tabular}{cccccc}
        \toprule
        \textbf{Data Graph}                                         & $|V|$ & $|E|$ & \includegraphics[width=1em]{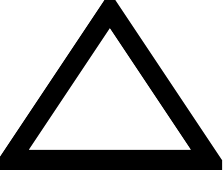} &
        \includegraphics[width=1em]{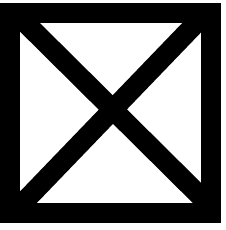} &
        \includegraphics[width=1em]{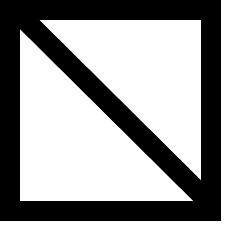}                                                                                            \\
        \midrule
        as-Skitter (as) \cite{snapnets}                             & 1.7E6 & 1.1E7 & 2.9E7                                                        & 1.5E8  & 2.0E9  \\
        LiveJournal (lj) \cite{snapnets}                            & 4.8E6 & 4.3E7 & 2.9E8                                                        & 9.9E9  & 7.6E10 \\
        Orkut (ok) \cite{snapnets}                                  & 3.1E6 & 1.2E8 & 6.3E8                                                        & 3.2E9  & 6.7E10 \\
        uk-2002 (uk) \cite{LAW-Datasets}                            & 1.8E7 & 2.6E8 & 4.4E9                                                        & 1.6E11 & 2.7E12 \\
        FriendSter (fs) \cite{snapnets}                             & 6.5E7 & 1.8E9 & 4.2E9                                                        & 9.0E9  & 1.8E11 \\
        \bottomrule
    \end{tabular}
\end{table}

The BFS-style algorithms decompose the pattern graph recursively into a series of join units.
A join unit is a simple partial pattern graph whose matching results can be enumerated easily from the data graph or a pre-computed index.
The BFS-style algorithms enumerate matching results of join units first and assemble them via one or more rounds of joining to get the matching results for the whole pattern graph.
The algorithms shuffle the partial matching results (intermediate results) during the join.
Researchers propose varieties of join units (Edge \cite{BiGJoin}, Star \cite{SEED} \cite{PSgL} , TwinTwig \cite{TwinTwigJoin} \cite{TwinTwigJoinJournalVersion}, Clique \cite{SEED}, Crystal \cite{CBF}) and join frameworks (Left-deep join \cite{TwinTwigJoin,TwinTwigJoinJournalVersion}, Bushy join \cite{SEED}, Hash-assembly \cite{CBF}, Generic join \cite{BiGJoin}) to reduce intermediate results.

However, BFS-style algorithms are still costly.
First, shuffling partial matching results is inevitable in the join-based framework, causing high communication costs.
The typical pattern graphs in Table \ref{tab:statistics_of_typical_pattern_graphs} are the core structures of many complex pattern graphs in Fig.\ref{fig:pattern_graphs}.
Just shuffling matching results of the core structures will cause high communication costs.
Second, some cutting-edge algorithms like SEED \cite{SEED} and CBF \cite{CBF} build extra index structures like SCP index (in SEED) or clique index (in CBF) for each data graph to achieve high performance.
The index requires non-trivial computation costs to construct and store.
It also requires extra costs to maintain if the data graph is dynamic, which is common in the industry.

The drawbacks of the existing methods inspire us designing a new distributed (continuous) subgraph enumeration framework that 1) avoids shuffling partial matching
results, 2) does not rely on any extra index, and 3) scales to large data graphs and complex pattern graphs.

\subsection{Contributions}

In this work, we present two new distributed Backtracking-based subgraph ENUmeration (BENU) frameworks\footnote{The two frameworks are open sourced at \url{https://github.com/PasaLab/BENU}.}: the Batch-BENU framework for static data graphs and the Streaming-BENU framework for dynamic data graphs.
An earlier version of this work \cite{BENU-ICDE} was presented at the 35th IEEE International Conference on Data Engineering (ICDE 2019).
In that version, we proposed the Batch-BENU framework for distributed subgraph enumeration in \emph{static} undirected data graphs and implemented it with MapReduce.
Batch-BENU does not shuffle intermedia results or use indices.
Instead, it stores the data graph in a distributed database and queries the adjacency sets of the data graph on demand, driven by backtracking-based execution plans.

However, simply extending Batch-BENU to process \emph{dynamic} data graphs is inefficient.
Batch-BENU has to enumerate subgraphs in the latest data graph snapshot at every time step and compare the matching results with the previous time step to detect appearing/disappearing subgraphs (i.e., incremental matches).
Enumerating subgraphs from scratch repeatedly contains lots of redundant computation.
The long execution time of subgraph enumeration can hardly meet the near real-time performance requirement of online applications.
Supporting dynamic data graphs is not a trivial extension.
The challenge is reducing redundant computation as much as possible and enumerating incremental matches directly from the update of each time step.

To achieve the target, the cutting-edge methods for dynamic data graphs either maintain the latest matching results in memory \cite{SJ-Tree} \cite{TurboFlux} \cite{D-IDS} or have to eliminate contradictory results with extra shuffling \cite{BiGJoin}.
They are inefficient in storage and communication, respectively.

In this work, we propose a novel concept--\emph{incremental pattern graphs}--for continuous subgraph enumeration in \emph{dynamic} graphs.
We prove that finding incremental matches at every time step is equivalent to enumerating  subgraph instances of incremental pattern graphs.
We propose the Streaming-BENU framework to enumerate them from the updated edges directly, guided by backtracking-based incremental execution plans.
Streaming-BENU outputs valid and duplication-free results without maintaining any matching result in memory.

Overall, we make the following contributions.

First, we propose a distributed subgraph enumeration framework \emph{Batch-BENU}.
Batch-BENU generates local search tasks for every data vertex and executes the tasks in parallel on a distributed computing platform.
A local search task enumerates matches of the pattern graph in the local neighborhood of a data vertex, following a backtracking-based execution plan.
Batch-BENU does not shuffle any partial matching result or use any index.
Instead, it queries the data graph stored in a distributed database \emph{on demand}.

Second, we propose a search-based method to generate the best execution plan.
The method includes three execution plan optimization techniques (common subexpression elimination, instruction reordering, and triangle caching), a cost estimation model, and two pruning techniques.

Third, we propose the concept of incremental pattern graphs to support continuous subgraph enumeration in dynamic graphs.
Based on the concept, we solve continuous subgraph enumeration by enumerating matches of incremental pattern graphs in the data graph snapshots at each time step.
We develop the \emph{Streaming-BENU} framework to enumerate their matches efficiently.

Forth, we propose efficient implementations of Batch-BENU and Streaming-BENU.
We propose the local database technique to reduce communication costs and the task splitting technique to balance workloads.
The experimental results validate the efficiency and scalability of the two frameworks.
They outperform the state-of-the-art methods on complex pattern graphs by up to one and two orders of magnitude, respectively.

We organize the rest of the paper as follows.
Section 2 defines the problem and introduces related techniques.
Section 3 describes the Batch-BENU framework.
Section 4 presents the method to generate the best execution plan for Batch-BENU.
Section 5 elaborates on the Streaming-BENU framework for dynamic graphs.
Section 6 discusses the implementations.
Section 7 experimentally evaluates their performance.
Section 8 briefly surveys the related work.
Section 9 concludes the work.

\section{Preliminaries}

We first define the problem of subgraph enumeration and its continuous variant.
Then, we introduce the backtracking-based framework for subgraph matching.

\subsection{Problem Definition}

In this work, we focus on processing simple unlabeled graphs.
We define a \emph{static} graph $g$ as $g=(V(g),E(g))$, where $V(g)$/$E(g)$ is the vertex/edge set of $g$.
If $g$ is \emph{undirected}, we denote the adjacency set of a vertex $v$ as $\Gamma_{g}(v)=\{w|(w,v)\in E(g)\}$.
The degree of $v$ is $d_{g}(v)=|\Gamma_{g}(v)|$.
If $g$ is \emph{directed}, we denote the incoming/outgoing adjacency set of a vertex $v$ as  $\Gamma_{g}^{\text{in}}(v)=\{w|(w,v)\in E(g)\}$/$\Gamma_{g}^{\text{out}}(v)=\{w|(v,w)\in E(g)\}$.
A subgraph $g'$ of $g$ is a graph such that $V(g')\subseteq V(g)$ and $E(g')\subseteq E(g)$.
An \emph{induced subgraph} $g(V')$ of a graph $g$ on a vertex set $V'$ is defined as $g(V')=(V'\cap V(g),\{(u,w)|(u,w)\in E(g), u\in V', w\in V'\})$.

The subgraph enumeration involves two graphs: a data graph $G$ and a pattern graph $P$.
Let $N=|V(G)|$, $M=|E(G)|$, $n=|V(P)|$ and $m=|E(P)|$.
The pattern graph $P$ is usually much smaller than $G$, i.e., $n\ll N,m\ll M$.
We assume $P$ is connected.
We use $v_{i}$/$u_{i}$ to denote a vertex from the data/pattern graph.
Without loss of generality, we assume that vertices in $G$ and $P$ are consecutively numbered, i.e., $V(G)=\{v_{1},v_{2},\dots,v_{N}\}$ and $V(P)=\{u_{1},u_{2},\dots,u_{n}\}$.
A match of $P$ in $G$ is defined in Definition~\ref{def:match}.
An isomorphic subgraph of $P$ in $G$ is defined in Definition~\ref{def:isomorphic_subgraph}.
Taking Fig.~\ref{fig:demo_case_of_subgraph_enumeration} as the example, the subgraph shown with bold lines in $G$ is isomorphic to $P$ with a match $f'=(v_{1},v_{2},v_{3},v_{4},v_{5},v_{8})$.

\begin{definition}[Match]\label{def:match}
	Given a pattern graph $P$ and a data graph $G$, a mapping $f:V(P) \rightarrow V(G)$ is a \emph{match} of $P$ in $G$ if $f$ is \emph{injective} and $\forall x \forall y: (x,y)\in E(P)\rightarrow (f(x),f(y))\in E(G)$.
	A match $f$ is denoted as $f=(f_{1},f_{2},\dots,f_{n})$, where $f_{i}=f(u_{i})$ for $1\leq i\leq n$.
\end{definition}

\begin{definition}\label{def:isomorphic_subgraph}
	Given a pattern graph $P$ and a data graph $G$ , a subgraph $g$ of $G$ is isomorphic to $P$ if and only if there exists a match $f$ of $P$ in $g$, $|V(P)|=|V(g)|$ and $|E(P)|=|E(g)|$.
\end{definition}

\begin{figure}[!t]
	\centering
	\subfloat[Pattern graph $P$\label{fig:demo_pattern_graph}]{~~~~\includegraphics[height=1.8cm]{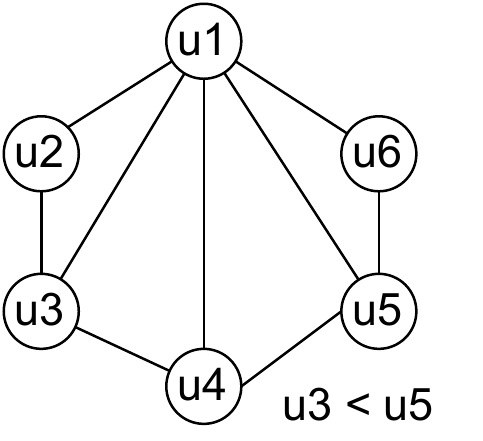}~~~~}
	\hfil
	\subfloat[Data graph $G$\label{fig:demo_data_graph}]{\includegraphics[height=1.8cm]{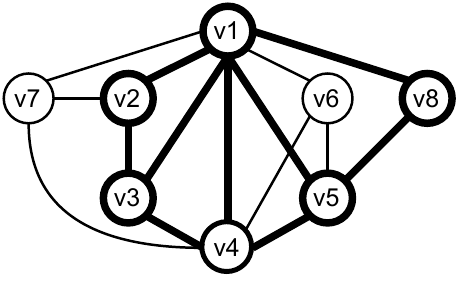}}
	\caption{Toy case of subgraph enumeration.\label{fig:demo_case_of_subgraph_enumeration}}
\end{figure}

We follow \cite{SEED} to define the subgraph enumeration problem in Definition~\ref{def:subgraph_enumeration}.
We denote the set of the isomorphic subgraphs of $P$ in $G$ as $R_{G}(P)$.
The task of subgraph enumeration is to calculate $R_{G}(P)$.
Subgraph enumeration focuses on undirected $P$ and $G$.
We extend to directed ones in continuous subgraph enumeration.

\begin{definition}\label{def:subgraph_enumeration}
	Given a \emph{static} undirected pattern graph $P$ and a \emph{static} undirected data graph $G$, the task of \emph{subgraph enumeration} is to enumerate all subgraphs of $G$ that are isomorphic to $P$.
\end{definition}

When data graphs are \emph{dynamic}, the subgraph enumeration becomes the \emph{continuous} subgraph enumeration.
In a dynamic graph, vertices and edges are inserted to/removed from the graph in a streaming manner.
Since a vertex insertion/removal operation can be decomposed into multiple edge insertion/removal operations, we focus on handling dynamic graphs with edge updates.

A dynamic data graph is defined as $G'=(V(G'_{0}),E(G'_{0}),\Delta G')$, where $V(G'_{0})$/$E(G'_{0})$ is the initial vertex/edge set of $G'$.
$\Delta G'$ is the update stream of $G'$.
$\Delta G'=\{\Delta o_{1},\Delta o_{2},\dots\}$ consists of a sequence of batch updates $\Delta o_t$.
$\Delta o_{t}=\{(op_{1}, v_{j_1}, v_{k_1}), (op_2, v_{j_2}, v_{k_2}), \dots\}$ consists of inserting and deleting edges between time step $t$ and $t-1$ ($t \geq 1$).
$op_i$ can be $+$ or $-$, indicating inserting or deleting the edge $(v_{j_i},v_{k_i})$ to or from $G'$.
We assume that an edge appears at most once in $\Delta o_{t}$, either inserted or deleted.
By applying $\Delta o_{1}, \Delta o_{2}, \dots, \Delta o_{t}$ to the initial graph $G'_0$ of $G'$, we can get the snapshots of $G'$ $G'_1, G'_2, \dots, G'_t$ in turn.
We use $E(G'_t)$ to denote the edge set of $G'$ at time step $t$.
Fig.~\ref{fig:demo_case_of_continuous_subgraph_enumeration} shows a demo dynamic data graph $G'$ and its snapshots.
$G'$ is directed.
The inserting/deleting edges in $\Delta o_{t}$ are listed below the arrow.
The solid blue edges in $G'_{t}$ are inserting edges while the faded dotted edges are deleting edges.

We use $R_{G'_{t}}(P)$ to denote the set of isomorphic subgraphs of $P$ in $G'_{t}$.
The target of continuous subgraph enumeration is to detect changes in $R_{G'_{t}}(P)$ and report incremental matches $\Delta R_{t}^{+}$ and $\Delta R_{t}^{-}$ as defined in Definition~\ref{def:continuous_subgraph_enumeration}.
We assume that the batch size $|\Delta o_{t}|$ is much smaller than $|E(G'_t)|$.
Taking the demo case in Fig.~\ref{fig:demo_case_of_continuous_subgraph_enumeration} as the example, the output of each time step is shown in two rows.

Table~\ref{tab:notations} summarizes the frequently used notations in this work.

\begin{definition}\label{def:continuous_subgraph_enumeration}
	Given a \emph{static} pattern graph $P$ and a \emph{dynamic} data graph $G'$, the task of \emph{continuous subgraph enumeration} is to report appearing matches $\Delta R_{t}^{+} = R_{G'_{t}}(P) \backslash R_{G'_{t-1}}(P)$ and disappearing matches $\Delta R_{t}^{-} = R_{G'_{t-1}}(P)\backslash R_{G'_{t}}(P)$ for every time step $t$ ($t \geq 1$), where $\backslash$ is the set difference operator.
\end{definition}

\begin{table}[!t]
	\centering
	\caption{Notations\label{tab:notations}}
	\begin{tabular}{cp{19em}}
		\toprule
		\textbf{Notation}                                                 & \textbf{Description}                                                                  \\
		\midrule
		$G,G_{t}$                                                         & The data graph $G$. If $G$ is dynamic, $G_{t}$ is the snapshot at time step $t$.      \\
		$P,n,m$                                                           & The pattern graph $P$. $n=|V(P)|$. $m=|E(P)|$.                                        \\
		$\Delta P_{i}$                                                    & The $i$-th incremental pattern graph of $P$.                                          \\
		$u,u_{i}$                                                         & An arbitrary/The $i$-th vertex in $P$.                                                \\
		$v,v_{i}$                                                         & An arbitrary/The $i$-th vertex in $G$.                                                \\
		$\Gamma_{g}(x),\Gamma_g^{\text{in}}(x), \Gamma_g^{\text{out}}(x)$ & The (incoming/outgoing) adjacency set of the vertex $x$ in the graph $g$.             \\
		$f=(f_{1},...,f_{n})$                                             & A match $f$ of $P$ in $G$. $f_{i}=f(u_{i})$.                                          \\
		$R_{G}(P)$,$R_{G_{t}}(P)$                                         & The set of matches of the pattern graph $P$ in the data graph $G$ (snapshot $G_{t}$). \\
		$\Delta R_{t}^{+}$,$\Delta R_{t}^{-}$                             & The appearing/disappearing matches of the pattern graph at time step $t$.             \\
		\bottomrule
	\end{tabular}
\end{table}

\subsection{Symmetry Breaking}

A match $f$ of the pattern graph $P$ in the data graph $G$ (snapshot $G_{t}$) corresponds to a subgraph $g$ isomorphic to $P$ in $G$ ($G_{t}$).
However, multiple matches may correspond to the same subgraph due to the automorphism in $P$.
In Fig.~\ref{fig:demo_case_of_subgraph_enumeration}, the match $f'=(v_{1},v_{2},v_{3},v_{4},v_{5},v_{8})$ and $f''=(v_{1},v_{8},v_{5},v_{4},v_{3},v_{2})$ both correspond to the subgraph $g$ shown with bold lines in $G$.
Enumerating all matches of $P$ in $G$ may report duplicate subgraphs.

We adopt the \emph{symmetry breaking} technique \cite{SymmetryBreaking} to avoid such duplication.
The technique requires a total order $\prec$ defined on $V(G)$.
It also imposes a partial order $<$ on $V(P)$.
The technique redefines a match $f$ of $P$ in $G$ as a mapping satisfying both Definition~\ref{def:match} and the partial order constraints: if $u_{i} < u_{j}$ in $V(P)$, then $f(u_{i}) \prec f(u_{j})$ in $V(G)$.
Under the new definition, if a subgraph $g$ is isomorphic to $P$, there is one and only one match $f$ of $P$ in $g$ \cite{SymmetryBreaking}.
It establishes a \emph{bijective} mapping between matches of $P$ in $G$ and isomorphic subgraphs of $P$ in $G$.
In Fig.~\ref{fig:demo_case_of_subgraph_enumeration}, the partial order imposed on $P$ is $u_{3} < u_{5}$.
Assuming $v_{3}\prec v_{5}$ in the total order, the subgraph $g$ shown with bold lines in $G$ is isomorphic to $P$ with only one match $f'=(v_{1},v_{2},v_{3},v_{4},v_{5},v_{8})$.

We take advantage of the technique to convert the problem of enumerating subgraphs into enumerating matches.
In the following sections, we use matches to represent isomorphic subgraphs interchangeably.
For static data graphs $G$, we use the degree-based total order $\prec$ defined in \cite{SEED}.
For dynamic data graphs, we use the natural order of vertex IDs as the total order.

\subsection{Backtracking-based Framework}

The backtracking-based framework is popular among serial subgraph isomorphism algorithms.
It incrementally maps each pattern vertex to data vertices in the match $f$ according to a given matching order $O$.
Algorithm \ref{alg:backtracking_based_framework} shows a simplified version of the original framework \cite{BacktrackingFramework}.

The \noun{SubgraphSearch} procedure finds all the matches of $P$ in $G$ recursively.
The \noun{NextPatternVertexToMatch} function returns the next unmapped pattern vertex $u_{i}$ in $f$ according to the matching order $O$.
The \noun{RefineCandidates} function calculates a candidate set $C_{i}$ of the data vertices that we can map $u_{i}$ to.
Mapping $u_{i}$ to any data vertex in $C_{i}$ should not break the match conditions in Definition~\ref{def:match} and the partial order constraints.
The framework recursively calls \noun{SubgraphSearch} until all vertices are mapped in $f$. Different algorithms have different implementations for \noun{FirstPatternVertexToMatch}, \noun{NextPatternVertexToMatch}, and \noun{RefineCandidates}.

\begin{algorithm}[!t]
	\small
	\textbf{Input:} Pattern graph $P$, Matching order $O$, Data graph $G$.
	\begin{algorithmic}[1]
		\State $f \leftarrow $ an empty mapping from $V(P)$ to $V(G)$;
		\State $u_i \leftarrow \text{\noun{FirstPatternVertexToMatch}}(O)$;
		\ForAll{$v_j \in V(G)$}
		\State $f_i \leftarrow v_j$;
		\State \Call{SubgraphSearch}{$P,G,O,f$};
		\EndFor
		\Procedure{SubgraphSearch}{$P,G,O,f$}
		\If{all pattern vertices are mapped in $f$} output $f$;
		\Else
		\State $u_i \leftarrow \text{\noun{NextPatternVertexToMatch}}(O,f)$;
		\State $C_i \leftarrow \text{\noun{RefineCandidates}}(P,G,f,u_i)$;
		\ForAll{$v_k \in C_i$}
		\State $f_i \leftarrow v_k$;
		\State \Call{SubgraphSearch}{$P,G,O,f$};
		\State $f_i \leftarrow \text{NULL}$; \Comment{Make $u_i$ unmapped in $f$}
		\EndFor
		\EndIf
		\EndProcedure
	\end{algorithmic}

	\caption{Backtracking-based Framework\label{alg:backtracking_based_framework}}

\end{algorithm}

\section{Batch-BENU Framework}

We consider the shared-nothing cluster as the target distributed environment.
Each machine in the cluster has a limited memory that may be smaller than the data graph.
The approaches like \cite{QFrag} that load the whole data graph in memory are not feasible here.

\subsection{Framework Overview}

The DFS-style distributed subgraph enumeration method \cite{UllmanOneRound} is not efficient because of its one-round shuffle design.
It blindly shuffles edges before enumeration and cannot exploit the information of partial matching results.
Consider a special case where the data graph has no triangle but the pattern graph has.
A more efficient way than one-round shuffle is to try enumerating triangles first and then stop immediately after finding there is no triangle.

It inspires us to propose the \emph{on-demand shuffle} technique.
The main idea is to store the edges of the data graph in a distributed database and query (``shuffle'') the edges as needed during enumeration.
The technique follows the backtracking-based framework in Algorithm~\ref{alg:backtracking_based_framework} to enumerate matches.
Only when the framework needs to access the data graph in the \noun{RefineCandidates} function, it queries the database.
Once a partial match $f$ fails in the search that generates an empty candidate set for a pattern vertex, the framework skips $f$ and backtracks, not wasting any effort on mapping other pattern vertices in $f$.
By this way, the technique avoids querying useless edges.
It also avoids shuffling any partial matching result.

Around the on-demand shuffle technique, we develop the Batch-BENU framework (\emph{B-BENU}, for short) for distributed subgraph enumeration.
Algorithm~\ref{alg:BENU_framework} shows its workflow.
B-BENU stores the data graph $G$ in a distributed database $DB$ in parallel (Line 1).
Given a pattern graph $P$, B-BENU computes its best execution plan $E$ to enumerate the pattern graph $P$ on the master node (Line 2) and broadcasts $E$ and $P$ to worker nodes (Line 3).
The execution plan is a core concept in B-BENU.
An execution plan follows the backtracking-based framework to enumerate matches of $P$.
It gives out the matching order and detailed steps to calculate the candidate set for every pattern vertex.
We elaborate on it later.
B-BENU generates a local search task for each data vertex $v$ in $V(G)$ (Line 4).
$v$ is the \emph{starting} vertex of the local search task.
B-BENU executes all tasks in parallel with a distributed computing platform.
A local search task enumerates matches of $P$ in the neighborhood around the starting vertex $start$ (Line 5 to Line 8).
It initializes an empty mapping $f$ and maps the first pattern vertex $u_j$ in the matching order to $start$ (Line 6).
A local search task then follows the execution plan $E$ to match the remaining pattern vertices in $f$.
During the task execution, the execution plan queries the database on demand.

\begin{algorithm}[!t]
    \small
    \textbf{Input:} Pattern graph $P$, Data graph $G$, Database $DB$
    \begin{algorithmic}[1]
        \State Store $G$ in $DB$; \Comment{Initialization independent of $P$}
        \State $E \leftarrow \text{\noun{GenerateBestExecutionPlan}}(P)$;
        \State Broadcast $P$ and $E$ to worker machines;
        \ForAll{$start \in V(G)$} \emph{in parallel} \Comment{Local search tasks}
        \State $f \leftarrow $ an empty mapping from $V(P)$ to $V(G)$;
        \State $u_j \leftarrow \text{\noun{FirstPatternVertexToMatch}}(E)$;
        \State $f_j \leftarrow start$;
        \State Match the remaining pattern vertices in $f$ guided by $E$;
        \EndFor
    \end{algorithmic}

    \caption{Batch-BENU Framework\label{alg:BENU_framework}}

\end{algorithm}

\subsection{Execution Plan}

The execution plan gives out the detailed steps to enumerate matches of $P$ in $G$.
It is the core of the B-BENU framework.
Since a database querying operation is expensive due to its high latency, the execution plan queries the database on the level of adjacency sets instead of edges, to reduce the number of database operations.
The execution plan implements the three core functions in Algorithm~\ref{alg:backtracking_based_framework} as:

\noun{FirstPatternVertexToMatch} and \noun{NextPatternVertexToMatch}: Each execution plan is bound with a static matching order $O$. The functions return the first pattern vertex in $O$ that is \emph{unmapped} in the partial match $f$ as the first/next vertex to match.

\noun{RefineCandidates}: The execution plan intersects adjacency sets to calculate the candidate set $C_{j}$ for an unmapped pattern vertex $u_{j}$.
$\mathcal{N}(u_j)$ is the set of $u_j$'s neighbors in $P$ that are before $u_j$ in $O$.
$\mathcal{N}(u_j) = \{u_x | u_x \in \Gamma_P(u_j), u_x \text{ is before } u_j \text{ in } O\}$.
The pattern vertices in $\mathcal{N}(u_j)$ are already mapped in $f$ when we calculate $C_j$.
If $\mathcal{N}(u_j) = \emptyset$, $C_j = V(G)$.
Otherwise, for any vertex $u_i \in \mathcal{N}(u_j)$, if we map $u_j$ to $v_x$ in $f$, $v_x$ should be adjacent to $f_i$ in $G$.
In other words, $C_{j} = \bigcap_{u_i \in \mathcal{N}(u_j)}{\Gamma_G(f_i)}$.
Mapping $u_{j}$ to any vertex outside $C_{j}$ will violate the match condition $(u_{j},u_{i}) \in E(P) \rightarrow (f_{j},f_{i})\in E(G)$.
$C_{j}$ is further filtered to ensure that the data vertices in it do not violate the injective condition and the partial order constraints.

\begin{example}
    In Fig. \ref{fig:demo_case_of_subgraph_enumeration}, assume the matching order is $O: u_{1}, u_{2}, \dots, u_{6}$.
    Suppose $u_{1}$ and $u_{2}$ are mapped in the partial match $f=(v_{1},v_{2},?,?,?,?)$.
    $u_3$ is the next pattern vertex to match.
    $\mathcal{N}(u_3) = \{u_1, u_2\}$.
    The candidate set $C_{3}$ is $C_{3} = \{v|v \in \Gamma_{G}(v_{1}) \cap \Gamma_{G}(v_{2}), v \neq v_{1}, v\neq v2 \}=\{v_{3},v_{7}\}$.
\end{example}

Integrating the core functions, we can get an abstract execution plan.
A demo execution plan for $P$ in Fig.~\ref{fig:demo_pattern_graph} is shown in Fig.~\ref{fig:demo_BENU_execution_plans}a.
The \emph{Filter} operation filters out data vertices not satisfying either the injective condition or the partial order constraints.
The demo execution plan is bound with the matching order $O:u_{1},u_{3},u_{5},u_{2},u_{6},u_{4}$, expressed by the order of loop variables.
Each loop corresponds to a recursive search level (Line 11 to 14) in Algorithm~\ref{alg:backtracking_based_framework}.
For ease of presentation, the recursion is expanded.

\begin{figure*}[!t]
    \centering
    \includegraphics[height=5.5cm]{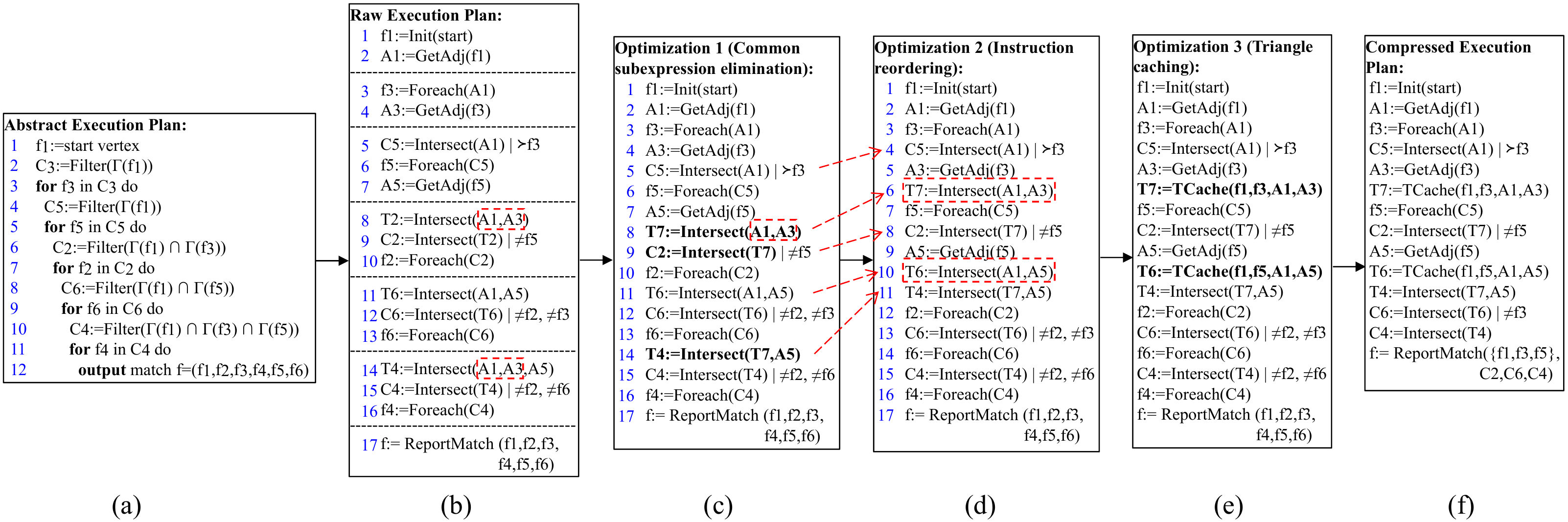}
    \caption{Batch-BENU execution plan and its optimizations for the toy pattern graph with $O:u_{1},u_{3},u_{5},u_{2},u_{6},u_{4}$.\label{fig:demo_BENU_execution_plans}}
\end{figure*}

\section{Execution Plan Generation}

In this section, we present the method to generate a concrete B-BENU execution plan for a given pattern graph $P$.
For a clear illustration, we use the same running example through the whole section.
The pattern graph is Fig.~\ref{fig:demo_pattern_graph} and the matching order is $O:u_{1},u_{3},u_{5},u_{2},u_{6},u_{4}$.
We first introduce how to generate a raw execution plan from a given matching order $O$.

\subsection{Raw Execution Plan Generation\label{subsec:raw_execution_plan_generation}}

Given a matching order $O:u_{k_{1}},u_{k_{2}},\dots,u_{k_{n}}$, the raw execution plan consists of a series of execution instructions.

\subsubsection{Execution Instruction}

A B-BENU execution instruction is denoted as
$$
    X:=Operation(Operands)[|FCs].
$$
It contains three parts: (1) a target variable $X$ that stores the result of the instruction, (2) an operation $Operation(Operands)$ describing the conducted operation and its operands, and (3) optional filtering conditions $FCs$.

There are 6 kinds of execution instructions in B-BENU as listed in Table~\ref{tab:types_of_execution_instructions}.
B-BENU uses two kinds of filtering conditions:
(1) a symmetry breaking condition, denoted as $\succ f_{i}$ or $\prec f_{i}$, means that vertices in $X$ should be bigger or smaller than $f_{i}$ under the total order $\prec$;
(2) an injective condition, denoted as $\neq f_{i}$, means that $f_{i}$ should be excluded from $X$.

\begin{table*}[!t]
    \begin{center}
        \caption{Types of Execution Instructions\label{tab:types_of_execution_instructions}}
        \begin{tabular}{p{12em}p{15em}p{29em}p{2em}}
            \toprule
            \textbf{Type}                          & \textbf{Operation}                   & \textbf{Description}                                                                                                        & \textbf{In}{*} \\
            \midrule
            Initialization (\textbf{INI})          & $f_{i}:=Init(start)$                 & Map $u_{i}$ to the starting vertex of the local search task in the partial match $f$.                                       & B,S            \\
            Database Querying (\textbf{DBQ})       & $A_{i}:=GetAdj(f_{i})$               & Get the adjacency set of the data vertex $f_{i}$ from the database.                                                         & B              \\
            Set Intersection (\textbf{INT})        & $X:=Intersect(\dots)$                & Intersect the operands and store the result set in $X$.                                                                     & B,S            \\
            Enumeration (\textbf{ENU})             & $f_{i}:=Foreach(X)$                  & Map $u_{i}$ to the vertices in $X$ one by one in the partial match $f$ and enter the next level in the backtracking search. & B,S            \\
            Triangle Cache (\textbf{TRC})          & $X:=TCache(f_{i},f_{j},A_{i},A_{j})$ & Triangle enumeration with triangle cache.                                                                                   & B              \\
            Result Reporting (\textbf{RES})        & $f:=ReportMatch(f_{1},f_{2},...)$    & Successfully find a match$f$ of $P$ in $G$ (or snapshots) that maps $u_{i}$ to $f_{i}$.                                     & B,S            \\
            \midrule
            Database Querying (\textbf{DBQ})       & $X:=GetAdj(f_{i},ty ,di ,op)$        & Get the specified adjacency set of the data vertex $f_{i}$ in snapshots.                                                    & S              \\
            Delta Enumeration (\textbf{Delta-ENU}) & $op,f_{i}:=Foreach(X)$               & Map $u_{i}$ to the vertices in $X$ one by one in the partial match $f$, retrieve corresponding $op$, and enter the next
            level in the backtracking search.      & S                                                                                                                                                                                   \\
            In Set Test (\textbf{INS})             & $InSetTest(f_i, X)$                  & If $f_i$ is not in the set $X$, backtrack to the upper level.                                                               & S              \\
            \bottomrule
        \end{tabular}
    \end{center}
    \begin{flushleft}
        {\scriptsize *\emph{B}/\emph{S} indicates that the instruction is used in Batch-BENU/Streaming-BENU.}
    \end{flushleft}
\end{table*}

\subsubsection{Instruction Generation\label{sec:raw_execution_plan_generation}}

We generate instructions for each pattern vertex successively according to $O$.
We first generate two instructions for the first vertex $u_{k_{1}}$ in $O$: $f_{k_{1}}:=Init(start)$ and $A_{k_{1}}:=GetAdj(f_{k_{1}})$.
The two instructions prepare related variables for $u_{k_{1}}$.
For each of the remaining vertices $u_{k_{i}}$ in $O$ ($2 \leq i \leq n$), we generate the following instructions in sequence:

\begin{enumerate}

    \item $T_{k_{i}}:=Intersect(\dots)$.
          This INT instruction calculates the raw candidate set for $u_{k_{i}}$ by intersecting related adjacency sets.
          For any $u_{j}$ that is before $u_{k_{i}}$ in $O$ and adjacent to $u_{k_{i}}$ in $P$, we add $f_{j}$'s adjacency set $A_{j}$ as an operand of the instruction.
          If $u_{k_{i}}$ is not adjacent to any vertex before it in $O$, we add $V(G)$ as the operand.

    \item $C_{k_{i}}:=Intersect(T_{k_{i}})[|FCs]$.
          This INT instruction calculates the refined candidate set for $u_{k_{i}}$ by applying the filtering conditions.
          For any $u_{j}$ before $u_{k_{i}}$ in $O$, if $u_{j}$ and $u_{k_{i}}$ have a partial order constraint, the corresponding symmetry breaking condition is added.
          If $u_{j}$ and $u_{k_{i}}$ are not adjacent in $P$, an injective condition $\neq f_{j}$ is added.
          If $u_{j}$ and $u_{k_{i}}$ are adjacent, the injective condition can be omitted, since $T_{k_{i}} \subseteq A_{j}$, $f_{j} \notin A_{j}$ and thus $f_{j} \notin T_{k_{i}}$.

    \item $f_{k_{i}}:=Foreach(C_{k_{i}})$.
          This ENU instruction maps $u_{k_{i}}$ to the data vertices in $C_{k_{i}}$ one by one in the partial match $f$, and enters the next level in the backtracking search.

    \item $A_{k_{i}}:=GetAdj(f_{k_{i}})$.
          If there is any vertex $u_{j}$ that is adjacent to $u_{k_{i}}$ in $P$ and is after $u_{k_{i}}$ in $O$, $A_{k_{i}}$ will be used by a subsequent INT instruction to calculate the raw candidate set for $u_{j}$.
          In this case, we add a DBQ instruction to fetch $A_{k_{i}}$.
          Otherwise, we skip the instruction.
\end{enumerate}

Finally, we add the RES instruction to the execution plan.

After generating instructions, we conduct the \emph{uni-operand elimination}.
If an INT instruction has only one operand and no filtering condition like $T_{i}:=Intersect(X)$, we remove the instruction and replace $T_{i}$ with $X$ in the other instructions.
If an INT instruction generates a candidate set $C_x$ and $C_x$ will be output by the VCBC compression technique (introduced later in Section~\ref{sec:VCBC-compression}), we do not eliminate the instruction.
After eliminating all uni-operand instructions, we get the raw execution plan.

The raw execution plan is well-defined.
All the variables are defined before used.
It materializes the abstract execution plan as shown in Fig. \ref{fig:demo_BENU_execution_plans}a.
It can be converted to the actual code or be executed by an interpreter easily.

\begin{example}
    Fig. \ref{fig:demo_BENU_execution_plans}b shows the raw execution plan generated for the running example.
    The instructions generated for $u_{4}$ are the 14th to 16th instruction.
\end{example}

B-BENU supports integrating other filtering techniques like the degree filter by adding corresponding filtering conditions.
In practice, adding filtering conditions to the instructions nested by many ENU instructions should be very careful, since they may bring considerable overheads.

\subsection{Execution Plan Optimization\label{subsec:execution_plan_optimization}}

Though the raw execution is functional, it contains redundant computation.
We propose three optimizations to reduce it.

\subsubsection{Opt1: Common Subexpression Elimination}

We borrow the concept of \emph{common subexpression} from the programming analysis.
Some combinations of adjacency sets appear as operands in more than one INT instruction.
For example, the common subexpression $\{A_{1},A_{3}\}$ appears twice in the raw execution plan in Fig. \ref{fig:demo_BENU_execution_plans}b.
We should eliminate it as it brings redundant computation.

We use a frequent-item mining algorithm like Apriori to find all the common subexpressions with at least two adjacency sets.
We pick the subexpression with the most adjacency sets to eliminate.
If the two subexpressions have the same number of adjacency sets, we pick the more frequent one according to their appearances.
If they further have the same frequency, we pick the one appearing first.
We add an INT instruction $T_{j}:=Intersect(Subexpression)$ before the first instruction that the subexpression appears.
Here, $j$ is an unused variable index.
Then we replace the subexpression appeared in other INT instructions with $T_{j}$ to eliminate the redundancy.
We eliminate the common subexpressions repeatedly until there is no common subexpression.
Finally, we conduct uni-operand elimination.

\begin{example}
    In Fig.~\ref{fig:demo_BENU_execution_plans}b, $\{A_{1},A_{3}\}$ and $\{A_{1},A_{5}\}$ are both common subexpressions.
    According to the order, we pick $\{A_{1},A_{3}\}$ to eliminate.
    After replacing it with $T_{7}$ in Fig.~\ref{fig:demo_BENU_execution_plans}c, there is no other common subexpression and the optimization stops.
\end{example}

\subsubsection{Opt2: Instruction Reordering}

The position of the instruction in the execution plan significantly affects the performance.
If an instruction can be moved forward and nested by fewer ENU instructions, it will be executed many fewer times.
To optimize instruction positions, we reorder instructions in an execution plan with three steps.

First, \emph{flatten INT instructions}.
For an INT instruction that have more than two operands, we sort its operands according to their definition positions.
The operand defined earlier is in the front.
We flatten the instruction into a series of INT instructions with at most two operands.
For example, $T_{j}:=Intersect(A,B,C)$ can be flattened into two INT instructions $T_{j'}:=Intersect(A,B)$ and $T_{j}:=Intersect(T_{j'},C)$, where $j'$ is an unused variable index.
Flattening INT instructions does not affect the correctness of the execution plan, but it enables us to reorder set intersection operations in finer granularity.

Second, \emph{construct the dependency graph}.
The instructions in an execution plan have dependency relations among them.
For two instructions $I_{1}$ and $I_{2}$, if $I_{2}$ uses the target variable of $I_{1}$ in its operands or filtering conditions, then $I_{1}$ and $I_{2}$ have a dependency relation $I_{1}\rightarrow I_{2}$.
$I_{1}$ should always be before $I_{2}$, otherwise $I_{2}$ will use an undefined variable.
We construct a dependency graph to describe such dependency relations.
In the graph, instructions are vertices, and dependency relations are directed edges.
For example, Fig.~\ref{fig:dependency_graph} is the dependency graph of the execution plan in Fig.~\ref{fig:demo_BENU_execution_plans}c.
In Fig.~\ref{fig:dependency_graph}, we use the target variable to represent an instruction and we omit the RES instruction.

\begin{figure}[!t]
    \centering
    \includegraphics[height=1.8cm]{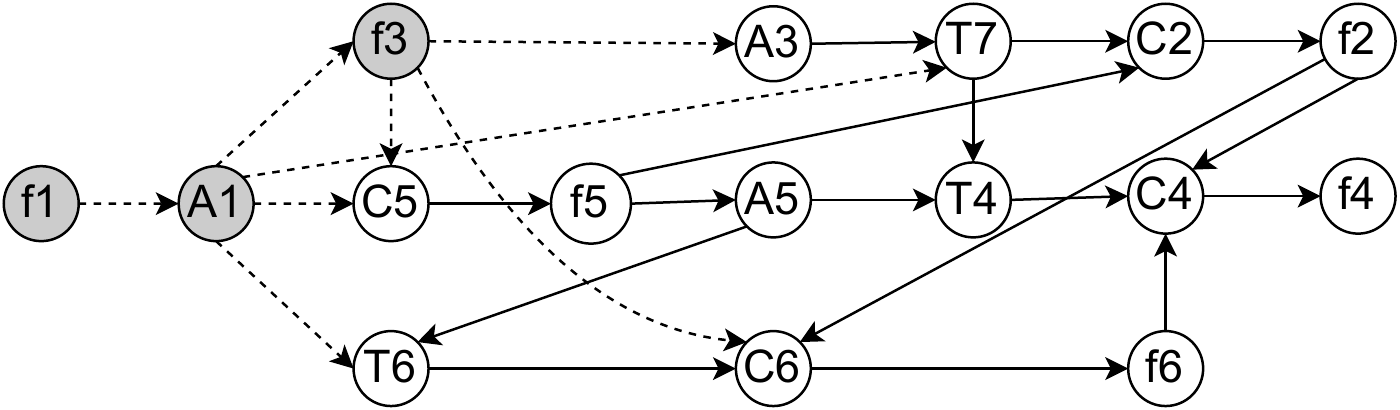}
    \caption{Dependency graph of the demo execution plan.\label{fig:dependency_graph}}
\end{figure}

Third, \emph{reorder instructions}.
We reorder the instructions by conducting topological sorting on the dependency graph.
The topological sort guarantees that the dependency relations between instructions are not violated.
During the sorting, it is common that several instructions can all be the candidate instructions for the next instruction.
For example, in Fig. \ref{fig:dependency_graph}, after sorting the first three instructions $[f_{1},A_{1},f_{3}]$, both $A_{3}$ and $C_{5}$ can be the next instruction under the topological order.
At this time, we rank the candidate instructions in an ascending order based on their instruction types: INI $<$ INT $<$ TRC $<$ DBQ $<$ ENU $<$ RES.
If two candidate instructions have the same type, the instruction in the front ranks higher.
This order guarantees that the INI and RES instructions must be the first and last instructions. The order of the other instructions is defined based on their execution costs.
The INT instructions are the cheapest as they only involve pure computation.
Moreover, if we can detect failed INT instructions that generate empty result sets earlier, we can stop the framework from wasting efforts on a doomed-to-fail partial match.
The TRC instructions involve cache accessing.
The DBQ instructions conduct database operations that are much more expensive than computation.
The ENU instructions are the most expensive as they add a level in the backtracking search and make the following instructions executed for more time.
We want to postpone them as much as possible.
The relative order of DBQ and ENU instructions reflects the matching order.
The ranking method also guarantees that the relative order is not changed.

\begin{example}
    For the execution plan in Fig.~\ref{fig:demo_BENU_execution_plans}c with its dependency graph in Fig.~\ref{fig:dependency_graph}, we can get a reordered execution plan in Fig.~\ref{fig:demo_BENU_execution_plans}d.
    The 14th instruction in Fig.~\ref{fig:demo_BENU_execution_plans}c is moved forward, crossing two ENU instructions $f_{2}$ and $f_{6}$.
\end{example}

\subsubsection{Opt3: Triangle Caching}

Suppose $u_{k_{1}}$ is the first vertex in the matching order.
If $u_{j}$ is a neighbor of $u_{k_{1}}$ in the pattern graph $P$, then $f_{k_{1}}$ and $f_{j}$ are neighbors in the data graph.
The INT instruction $X:=Intersect(A_{k_{1}},A_{j})$ calculates the vertices that can form a triangle with $f_{k_{1}}$ and $f_{j}$.
We find that some INT instructions in the execution plan repeatedly enumerate triangles around the starting vertex $f_{k_{1}}$.
For example, in Fig.~\ref{fig:demo_BENU_execution_plans}d, $T_{7}:=Intersect(A_{1},A_{3})$ and $T_{6}:=Intersect(A_{1},A_{5})$ both enumerate triangles around the starting vertex $f_{1}$.
Their computation is redundant.
The existing methods \cite{SEED} and \cite{CBF} avoid such redundancy by pre-enumerating triangles and storing them as an index.
The index requires non-trivial computation costs to maintain when the data graph is updated and occupies non-trivial disk space to store.

In B-BENU, we propose the triangle caching technique to reduce such redundancy \emph{on the fly}.
We set up a triangle cache for each local search task to cache the locally enumerated triangles.
For an INT instruction $X:=Intersect(A_{i},A_{j})$, if one of $f_{i}$ and $f_{j}$ is the starting vertex and the other one is its neighbor, we replace the INT instruction with a triangle caching instruction: $X:=TCache(f_{i},f_{j},A_{i},A_{j})$.
The triangle caching instruction queries the triangle cache with the key $[f_{i},f_{j}]$ first.
If the cache misses, it calculates $A_{i}\cap A_{j}$ and stores the result into the cache. Otherwise, it uses the pre-computed set in the cache as the result.

\begin{example}
    In Fig.~\ref{fig:demo_BENU_execution_plans}d, the marked instructions are replaced by the triangle caching instructions in Fig.~\ref{fig:demo_BENU_execution_plans}e.
\end{example}

\subsubsection{Support VCBC Compression\label{sec:VCBC-compression}}

The VCBC compression (vertex-cover based compression) \cite{CBF} is an efficient technique to compress the subgraph matching results based on a vertex cover $V_{c}$ of $P$.
Given a pattern graph $P$ and its vertex cover $V_{c}$, $core(P)$ is the induced subgraph of $P$ on $V_{c}$.
In VCBC, the matches of $core(P)$ in $G$ are helves.
For each helve, the matches of the pattern vertices not in $V_{c}$ are compressed in conditional image sets.
The helves and their conditional image sets form the compressed codes of the matching results of $P$ in $G$.

With modification, a B-BENU execution plan can directly output the VCBC-compressed matching results.
For an execution plan $E$ and a matching order $O$, assume the first $k$ pattern vertices in $O$ forms a vertex cover $V_{c}$ of $P$ while the first $k-1$ vertices do not.
The matches of the first $k$ pattern vertices are the helves.
For a pattern vertex $u_{j}$ not in $V_{c}$, we delete the ENU instruction of $f_{j}$ in $E$ and remove $f_{j}$ from the filtering conditions of other instructions.
We reserve the INT instruction that calculates the candidate set $C_{j}$ for $u_{j}$.
$C_{j}$ is equal to the conditional image set of $u_{j}$ in the VCBC compression.
We replace $f_{j}$ in the RES instruction with $C_{j}$ to directly output the compressed codes.

\begin{example}
    The execution plan in Fig.~\ref{fig:demo_BENU_execution_plans}e can be modified to Fig.~\ref{fig:demo_BENU_execution_plans}f to support the VCBC compression. The first three vertices $[u_{1},u_{3},u_{5}]$ in $O$ form the vertex cover $V_{c}$.
\end{example}

\subsubsection{Complexity Analysis}
The cost of optimizing a raw execution plan depends on the pattern graph $P$.
If the number of pattern vertices $n$ is fixed, the most expensive pattern graph to optimize is the $n$-clique, because it has the most edges and its raw execution plan has the most common subexpressions.
By inspecting the case of $n$-clique, we can get the worst-case computation complexity.

As for Optimization 1, an INT instruction in the raw execution plan has at most $n-1$ operands.
Any combination of the operands is a common subexpression.
The complexity of enumerating all common subexpressions in that instruction is $O(2^n)$.
Since there are $O(n)$ INT instructions, the complexity of enumerating common subexpressions in all instructions is $O(n2^n)$.
The complexity of eliminating a subexpression is $O(n^2)$.
The elimination will repeat $O(n)$ times until there is no common subexpression.
The worst-case time complexity of Optimization 1 is $O(n^22^n)$.

As for Optimization 2, the execution plan after flattening has $O(m)$ instructions.
Each instruction has at most 2 operands and $n$ injective conditions, so the dependency graph has $O(m)$ vertices and $O(nm)$ edges.
The complexity of topological sort is $O(nm)$.
If we use a heap to find the next instruction with the highest rank, the maintaince cost of the heap during the sort is $O(m\log m)$.
Therefore, the worst-case time complexity of Optimization 2 is $O(nm)$.

The costs of Optimization 3 and supporting VCBC compression are both linear to the number of instructions in the execution plan, which is $O(m)$.
Summarily, the computation complexity of the whole optimization is $O(n^22^n)$, dominated by Optimization 1.

\subsection{Best Execution Plan Generation\label{subsec:best-execution-plan-generation}}

Given a pattern graph $P$, any permutation of pattern vertices is a legal matching order.
Different matching orders correspond to different execution plans, having different execution costs.
We propose a search-based method to generate the best execution plan for a pattern graph.

\subsubsection{Execution Cost Estimation}

The execution cost of an execution plan $E$ is made up of the computation cost and the communication cost.
We define the computation cost as the number of executed times of all INT/TRC instructions.
We define the communication cost as the number of executed times of all DBQ instructions.
Thus, the core problem in estimating execution costs is to estimate numbers of executed times of instructions.

For an instruction, its number of executed times is related to the ENU instructions before it.
Assume the matching order of $E$ is $O:u_{k_{1}},u_{k_{2}},\dots,u_{k_{n}}$.
We denote the induced subgraph of $P$ on the first $i$ vertices in $O$ as the partial pattern graph $P_{i}$.
The leftmost column in Fig.~\ref{fig:backtracking_search_tree} shows the partial pattern graphs $P_{i}$ with the corresponding ENU instructions.
The pattern graph used in Fig.~\ref{fig:backtracking_search_tree} is Fig.~\ref{fig:demo_pattern_graph}.
The first $i$ ENU instructions actually enumerates matches of $P_i$ in $G$.
Thus, the number of executed times of the $i$-th ENU instruction is equal to the number of matches of $P_i$ in $G$.
The instructions between the $i$-th and $i$+1-th ENU instructions have the same number of executed times as the $i$-th ENU instruction.

We develop the \noun{EstimateComputationCost} function in Algorithm~\ref{alg:best_execution_plan_generation} to estimate the computation cost of an execution plan $E$.
The function tracks the partial pattern graph $p'$ as scanning instructions and uses the estimation model proposed in Section 5.1 of \cite{SEED} to estimate the number of matches of $p'$.
If $p'$ is disconnected, we decompose it into connected components and multiply the numbers of their matches together.
The estimation model can be replaced if a more accurate model is proposed later.

\begin{figure}[!t]
    \centering
    \includegraphics[width=7cm]{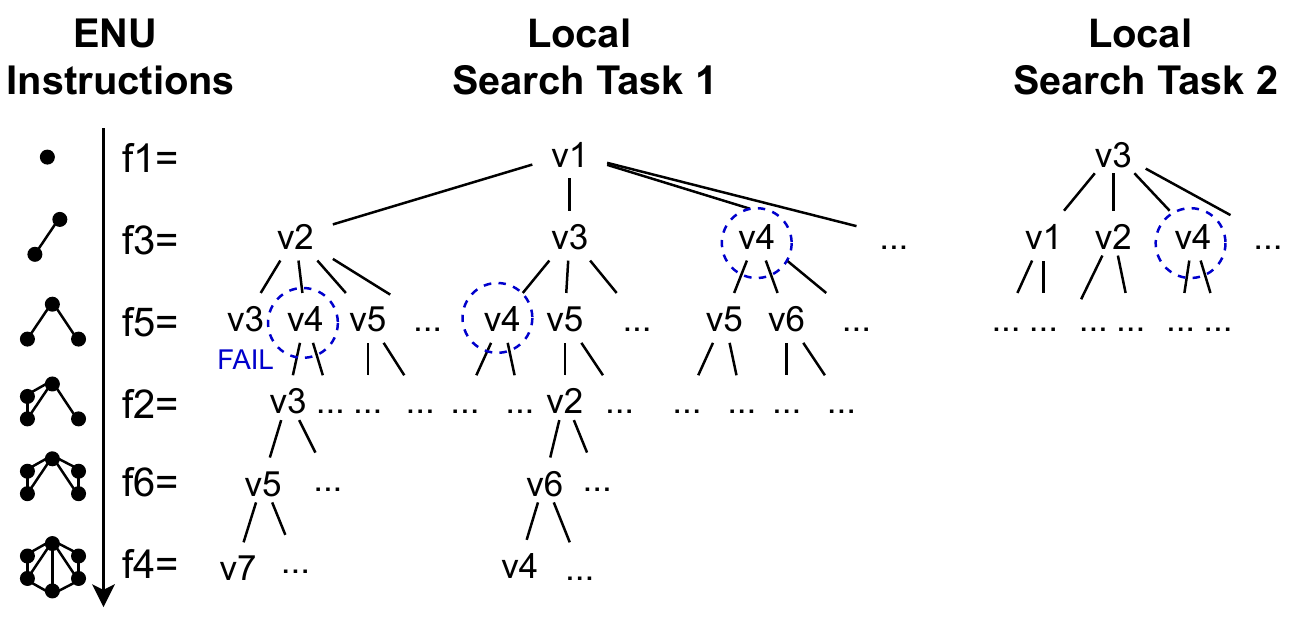}
    \caption{Backtracking search trees of local search tasks.\label{fig:backtracking_search_tree}}
\end{figure}

\subsubsection{Best Execution Plan Search\label{subsec:best_execution_plan_search}}

We define the best execution plan as the execution plan with the least communication cost, since executing a DBQ instruction consumes much more time than an INT/TRC instruction.
If several execution plans have the same least communication cost, we define the one with the least computation cost as the best.

We propose a search-based algorithm (Algorithm~\ref{alg:best_execution_plan_generation}) to find the best execution plan $E_{\mathit{best}}$ for a given pattern graph $P$.
The communication cost of an execution plan is determined by the relative order of DBQ and ENU instructions.
As the optimizations in Section~\ref{subsec:execution_plan_optimization} do not affect the relative order, the communication cost is solely determined by the matching order.
Thus, Algorithm~\ref{alg:best_execution_plan_generation} calls the \noun{Search} procedure to find the set of candidate matching orders $O_{\mathit{cand}}$ that have the least communication cost without actually generating execution plans.
For matching orders in $O_{\mathit{cand}}$, the algorithm generates optimized execution plans and picks the one with the least computation cost as $E_{\mathit{best}}$.

The \noun{Search} procedure uses backtracking to iterate all permutations of pattern vertices.
It maintains the unused pattern vertices in ${C}$ and recursively enumerates the next pattern vertex in the partial matching order $O$ from ${C}$ one by one.
To avoid blindly exploring all permutations, we propose two pruning strategies.

\paragraph*{Dual Pruning}

In Line 11, we use the dual condition to filter out redundant matching orders.
The dual condition is based on the syntactic equivalence (SE) relations \cite{BoostISO} between pattern vertices.
For two vertices $u_{i}$ and $u_{j}$ in $P$, $u_{i}$ is SE to $u_{j}$ (denoted as $u_{i}\simeq u_{j}$) if and only if $\Gamma_{P}(u_{i})-\{u_{j}\}=\Gamma_{P}(u_{j})-\{u_{i}\}$.
For example, in q4 of Fig.~\ref{fig:pattern_graphs}, $u_{1}\simeq u_{4}$ and $u_{2}\simeq u_{3}$.
Given two SE vertices $u_{i}\simeq u_{j}$ and a matching order $O$, we define the matching order got by swapping $u_{i}$ and $u_{j}$ in $O$ as its dual matching order $O'$.
The execution plans generated from $O$ and $O'$ have the same execution cost.
For an execution plan $E$ generated from $O$, if we swap $A_i$/$A_j$, $C_i$/$C_j$ and $f_i$/$f_j$ in every instruction and adjust the symmetry breaking conditions correspondingly in $E$ , we can get a dual execution plan $E'$ with the matching order $O'$.
$E'$ is correct because the candidate set calculation in $E'$ still follows the principle in Section \ref{sec:raw_execution_plan_generation}.
The partial pattern graphs $P_i$ and ${P'}_i$ induced by the first $i$ vertices in $E$ and $E'$ are isomorphic for any $1 \leq i \leq n$.
The execution times of the $i$-th ENU instructions in $E$ and $E'$ are same.
Therefore, the communication and computation costs of $E$ and $E'$ are same.
If $u_{i}\simeq u_{j}$ and $i<j$, we only need to keep the matching order that $u_{i}$ appears before $u_{j}$.

\paragraph*{Cost-based Pruning}

When Algorithm~\ref{alg:best_execution_plan_generation} searches candidate matching orders, it maintains the current partial pattern graph $p'$ and the partial communication cost \textit{commCost'} simultaneously in Line 13 to Line 18.
The cost is updated with two cases.
In case 1, at least one of $u$'s neighbors will appear after $u$ in $O$.
According to Section \ref{subsec:raw_execution_plan_generation}, a DBQ instruction will be generated for $u$.
The execution times of the instruction are equal to the number of matches of $p'$.
In case 2, no DBQ instruction will be generated.
The partial communication cost remains unchanged.
If the partial communication cost is already bigger than the current best cost, $O$ and all the orders \emph{expanded} from $O$ can be pruned safely.

The time complexity of the \texttt{Search} procedure is dominated by the estimation operation in Line 15.
The complexity of the operation is $O(m)$ and we denote its executed times as $\alpha$.
The time complexity of Line 4 to Line 7 is dominated by the optimized execution plan generation operation.
The complexity of the operation is $O(n^22^n)$ and we denote its executed times as $\beta$.
Therefore, the time complexity of Algorithm~\ref{alg:best_execution_plan_generation} is $O(\alpha m + \beta n^22^n)$.
$\alpha$ and $\beta$ are affected by the pattern graph.
The upper bound of $\alpha$ is $\sum_{i=1}^{n}{\mathcal{P}(n,i)}$ ($\mathcal{P}(n,i)$ is $i$-permutations of $n$).
The upper bound of $\beta$ is $n!$.
In practice, $\alpha$ and $\beta$ are much less than their upper bounds.

\begin{algorithm}[!t]
    \caption{Best Execution Plan Generation\label{alg:best_execution_plan_generation}}
    \small
    \newcommand{\var}[1]{\text{\textit{#1}}}
    \textbf{Input:} Pattern graph $P$.
    \textbf{Output:} Best execution plan $E_{\var{best}}$.
    \newcommand{\Owork}{O_{\var{work}}}
    \begin{algorithmic}[1]
        \State $E_{\var{best}} \leftarrow \text{NULL}$; $O_{\var{cand}} \leftarrow \{\}$; \Comment{Global variables}
        \State $\var{bCommCost} \leftarrow +\infty$; $\var{bCompCost} \leftarrow +\infty$; \Comment{Best costs}
        \State \Call{Search}{0, $V(P)$, new PartialPatternGraph(), [], 0}; \Comment{$O_\var{cand}$ is updated in \noun{Search}}

        \ForAll{$O \in O_{\var{cand}}$}
        \State $ E \leftarrow$ \noun{GenerateOptimizedExecutionPlan}($P$, $O$);
        \State $\var{cost} \leftarrow$ \Call{EstimateComputationCost}{$P$,$E$};

        \If{$\var{cost} < \var{bCompCost}$} $E_{\var{best}} \leftarrow E$, $\var{bCompCost} \leftarrow \var{cost}$;
        \EndIf
        \EndFor
        \State \Return{$E_{\var{best}}$}.

        \Procedure{Search}{$i, C, p, O, \var{commCost}$}
        \If{$i < |V(P)|$} \Comment{$O$ is not complete}
        \ForAll{$u \in C$ passing dual condition checking}
        \State $O[i] \leftarrow u$;
        ${C}' \leftarrow {C} - \{u\}$;
        \State $p' \leftarrow$ Add $u$ to the partial pattern graph $p$;
        \If{$\Gamma_P(u) \cap {C} \neq \emptyset$} \Comment{Case 1}
        \State $s \leftarrow$ Estimate the number of matches of $p'$;
        \Else \Comment{Case 2}
        \State $s \leftarrow 0$;
        \EndIf
        \State $\var{commCost}' \leftarrow \var{commCost} + s$;
        \If{$\var{commCost}' > \var{bCommCost}$} continue; \EndIf
        \State \Call{Search}{$i+1, {C}', p', O, \var{commCost}'$};
        \EndFor
        \Else \Comment{$O$ is complete}
        \If{$\var{commCost} < \var{bCommCost}$}
        \State $\var{bCommCost} \leftarrow \var{commCost}$; $O_{\var{cand}} \leftarrow \{O\}$;
        \ElsIf{$\var{commCost} = \var{bCommCost}$}
        \State $O_{\var{cand}} \leftarrow O_{\var{cand}} \cup \{O\}$.
        \EndIf
        \EndIf
        \EndProcedure

        \Function{EstimateComputationCost}{$P$, $E$}
        \State $\var{cost} \leftarrow 0$; $\var{curNum} \leftarrow 0$; $p' \leftarrow$ new PartialPatternGraph();
        \ForAll{instruction $I \in E$}
        \If{$I$.type is ENU}
        \State Update $p'$ with $I$;
        \State $\var{curNum} \leftarrow$ estimate the number of matches of $p'$;
        \ElsIf{$I$.type is INT or TRC}
        \State $\var{cost} \leftarrow \var{cost} + \var{curNum}$;
        \EndIf
        \EndFor
        \State \Return{\var{cost}}.
        \EndFunction
    \end{algorithmic}
\end{algorithm}

\section{Streaming-BENU Framework}

When data graphs are dynamic, the subgraph enumeration problem becomes the continuous subgraph enumeration problem.
A naive approach to the problem is conducting subgraph enumeration on $G'_{t}$ and $G'_{t-1}$ separately at each time step $t$ and calculating differences of $R_{G'_t}(P)$ and $R_{G'_{t-1}}(P)$.
However, enumerating subgraphs from scratch is expensive and contains redundant computation.
To overcome the drawback, we propose the Streaming-BENU framework that enumerates subgraphs in $\Delta R_{t}^{+}$ and $\Delta R_{t}^{-}$ incrementally from the batch update $\Delta o_{t}$.
For a clear illustration, we use the same example in Fig.~\ref{fig:demo_case_of_continuous_subgraph_enumeration} through the section.

\begin{figure*}[!t]
    \centering
    \includegraphics[height=5cm]{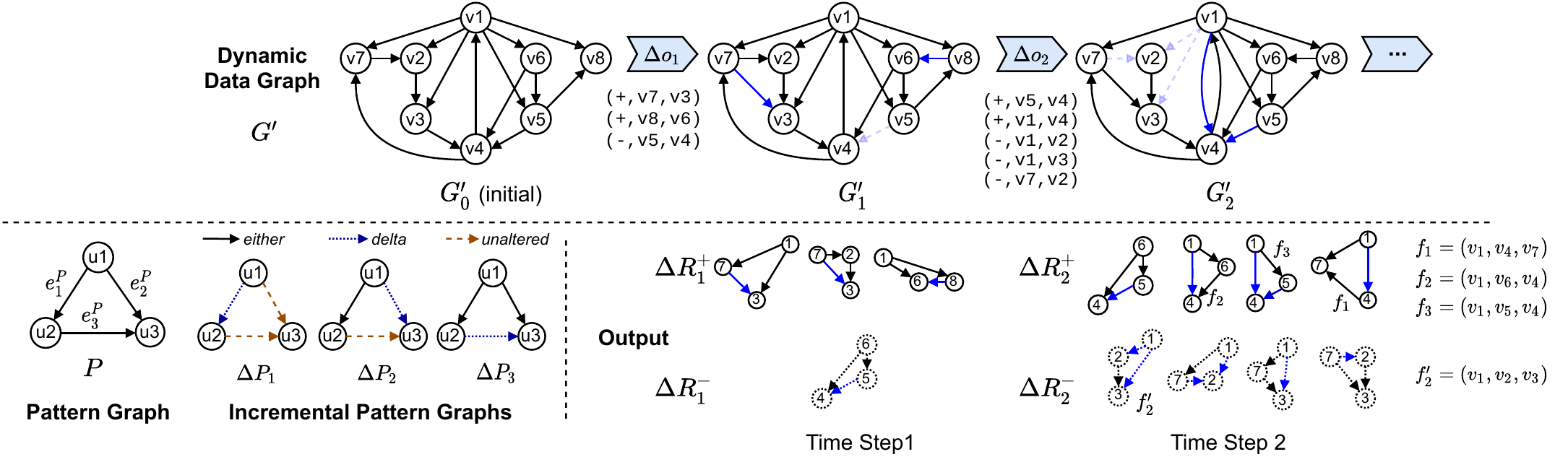}
    \caption{Demo case of continuous subgraph enumeration.\label{fig:demo_case_of_continuous_subgraph_enumeration}}
\end{figure*}

\subsection{Incremental Subgraph Matching}

Given a data graph $G'$ and a time step $t$ ($t \geq 1$), $G'_{t-1}$ and $G'_t$ are the snapshots related to $t$.
We can classify the edges $e$ of $G'_{t-1}$ and $G'_{t}$ into two types:
1) $e$ is a \emph{delta} edge if $e \in E(G'_{t}) \setminus E(G'_{t-1})$ (inserting edge) or $e \in E(G'_{t-1}) \setminus E(G'_t)$ (deleting edge);
2) $e$ is an \emph{unaltered} edge if $e \in E(G'_{t}) \cap E(G'_{t-1})$.
For the toy case in Fig.~\ref{fig:demo_case_of_continuous_subgraph_enumeration} at time step 2, $(v_{1},v_{4})$ is an inserting edge in $G'_2$, $(v_{1},v_{3})$ is a deleting edge in $G'_{1}$, and $(v_{4},v_{1})$ is an unaltered edge.

As stated by Theorem~\ref{thm:incremental_subgraph_must_contain_a_delta_edge}, any subgraph in the appearing matches $\Delta R_t^{+}$ or the disappearing matches $\Delta R_t^{-}$ must contain a delta edge.
It indicates that we only need to enumerate isomorphic subgraphs of $P$ that contain at least a delta edge.
Since $|\Delta o_t| \ll |E(G'_t)|$, the number of such subgraphs are much less than $|R_{G'_t}(P)|$.

\begin{theorem}\label{thm:incremental_subgraph_must_contain_a_delta_edge}
    For any $g \in \Delta R_t^+$, $g$ contains an inserting edge.
    For any $g \in \Delta R_t^-$, $g$ contains a deleting edge.
\end{theorem}
\begin{proof}
    We proof the theorem by contradiction.
    For any $g \in \Delta R_t^+$, we assume that $g$ does not contain any inserting edge.
    Since $g$ is an isomorphic subgraph of $P$ in $G'_t$, $g$ does not contain any deleting edge.
    $g$ only contains unaltered edges.
    $g$ is also a subgraph of $G'_{t-1}$.
    $g \in R_{G'_{t-1}}(P)$.
    It is inconsistent with $g \in \Delta R_t^+$ (i.e. $R_{G'_t}(P) \setminus R_{G'_{t-1}}(P)$).
    Therefore, for any $g \in \Delta R_t^+$, $g$ contains at least one inserting edge.
    The proof for $g \in \Delta R_t^-$ is similar.
\end{proof}

\subsubsection{Incremental Pattern Graph and Its Match}

To find the isomorphic subgraphs with delta edges, we extend a pattern graph $P$ with $m$ edges into $m$ \emph{incremental pattern graphs} $\Delta P_i$ ($1 \leq i \leq m$) as defined in Definition~\ref{def:incremental_pattern_graph}.
We number edges of $P$ consecutively.
Edge IDs are necessary to define the edge type mapping $\tau_i$ of $\Delta P_i$.
$\tau_i$ assigns every edge of $P$ to one of three types.
Fig.~\ref{fig:demo_case_of_continuous_subgraph_enumeration} shows the incremental pattern graphs of $P$.
$\tau_i$ is illustrated with edge colors.

\begin{definition}[Incremental Pattern Graph]\label{def:incremental_pattern_graph}
    Given a pattern graph $P$ with its edges numbered consecutively as $E(P)=\{e_{1}^{P},e_{2}^{P},\dots,e_{m}^{P}\}$, $P$ has $m$ incremental pattern graphs.
    The $i$-th incremental pattern graph (denoted as $\Delta P_{i}$) is a graph $\Delta P_{i}=(V(P),E(P),\tau_{i})$, where $\tau_i: E(P) \rightarrow \{\text{either, delta, unaltered}\}$ is an edge type mapping.

    $$
        \tau_{i}(e_{k}^{P})=\begin{cases}
            \text{either}    & 1 \leq k < i \\
            \text{delta}     & k = i        \\
            \text{unaltered} & i < k \leq m
        \end{cases}
    $$
\end{definition}

Definition~\ref{def:incremental_match} defines the \emph{incremental match} of an incremental pattern graph in $G'_t$ and $G'_{t-1}$ for every time step $t$.
An incremental match $f$ is a match of $P$ in the snapshot, but $f$ has type constraints on the data edges that a pattern edge can map to.
Definition~\ref{def:isomorphic_subgraph_of_incremental_pattern_graph} defines the isomorphic subgraph of an incremental pattern graph.
It is easy to see that a subgraph isomorphic to an incremental pattern graph is also isomorphic to the pattern graph.

\begin{definition}[Incremental Match]\label{def:incremental_match}
    Given a dynamic data graph $G'$, a pattern graph $P$, a time step $t$, and an incremental pattern graph $\Delta P_{i}$, a mapping $f: V(P) \rightarrow V(G')$ is an \emph{incremental match} of $\Delta P_{i}$ in $G'_{t}$ ($G'_{t-1}$) \emph{if and only if} $f$ satisfies:
    \begin{enumerate}
        \item $f$ is a match of $P$ in $G'_t$ ($G'_{t-1}$) satisfying the partial order constraints;
        \item For every edge $e^P=(s,t) \in E(P)$ and the data edge that $e^P$ is mapped to $e^{G'}=(f(s), f(t))$:

              If $\tau_i(e^P)=\text{either}$, $e^{G'}$ is an edge of $G'_t$ ($G'_{t-1})$;

              If $\tau_i(e^P)=\text{delta}$, $e^{G'}$ is a \emph{delta} edge of $G'_t$ ($G'_{t-1}$);

              If $\tau_i(e^P)=\text{unaltered}$, $e^{G'}$ is an \emph{unaltered} edge of $G'_t$ ($G'_{t-1}$).
    \end{enumerate}
\end{definition}

\begin{definition}\label{def:isomorphic_subgraph_of_incremental_pattern_graph}
    Given a dynamic data graph $G'$, a pattern graph $P$, a time step $t$, and an incremental pattern graph $\Delta P_{i}$, a subgraph $g$ of $G'_{t}$ ($G'_{t-1}$) is \emph{isomorphic} to $\Delta P_{i}$ \emph{if and only if} there exists an incremental match $f$ of $\Delta P_{i}$ in $g$, $|V(P)|=|V(g)|$ and $|E(P)|=|E(g)|$.
\end{definition}

\begin{example}
    For $t=2$ in Fig.~\ref{fig:demo_case_of_continuous_subgraph_enumeration}, $G'_2$ and $G'_1$ are related snapshots.
    $f_1$/$f_2$/$f_3$ is an incremental match of $\Delta P_1$/$\Delta P_2$/$\Delta P_3$ in $G'_2$, respectively.
    $f'_2$ is an incremental match of $\Delta P_2$ in $G'_1$.
    Their corresponding subgraphs in $G'_2$ and $G'_1$ are marked on the left.
    In $f_2$, $e_2^P$ is mapped to a delta edge $(v_1,v_4)$ of $G'_2$.
\end{example}


\subsubsection{Continuous Subgraph Enumeration via Subgraph Enumeration}

Given a dynamic graph $G'$, a pattern graph $P$ and a time step $t$, we denote the set of subgraphs isomorphic to $\Delta P_{i}$ in $G'_{t}$/$G'_{t-1}$ as $\Delta R_{t}^{i,+}$/$\Delta R_{t}^{i,-}$, respectively.
$\Delta R_{t}^{i,+}$ and $\Delta R_{t}^{i,-}$ have a strong connection with the output of continuous subgraph enumeration at each time step $t$ $\Delta R_t^+$ and $\Delta R_t^-$.

Theorem~\ref{thm:isomorphic_subgraph_must_be_matching_result} shows that every isomorphic subgraph of an incremental pattern graph in $G'_t$/$G'_{t-1}$ must be an appearing/disappearing match in $\Delta R_t^+$/$\Delta R_t^-$, respectively.
Theorem~\ref{thm:matching_result_must_be_isomorphic_subgraph} shows that every appearing/disappearing match in $\Delta R_t^+$/$\Delta R_t^-$ must a subgraph isomorphic to some incremental pattern graph in $G'_t$/$G'_{t-1}$, respectively.

\begin{theorem}\label{thm:isomorphic_subgraph_must_be_matching_result}
    For any $1\leq i\leq|E(P)|$, if $g\in\Delta R_{t}^{i,+}$, then $g\in\Delta R_{t}^{+}$; if $g\in\Delta R_{t}^{i,-}$, then $g\in\Delta R_{t}^{-}$.
\end{theorem}

\begin{proof}
    For any $i$ with $1 \leq i \leq |E(P)|$, for any $g \in \Delta R_{t}^{i,+}$, there is an incremental match $f$ of $\Delta P_i$ in $g$ according to Definition~\ref{def:isomorphic_subgraph_of_incremental_pattern_graph}.
    For the $i$-th edge $e_i^P=(s,t)$ of $P$, $\tau_i(e_i^P)=\text{delta}$.
    According to Definition~\ref{def:incremental_match}, $f$ maps $e_i^P$ to a delta edge $(f(s),f(t)) \in E(G'_t) \setminus E(G'_{t-1})$.
    $g$ cannot be a subgraph of $G'_{t-1}$.
    Thus, $g \notin R_{G'_{t-1}}(P)$.
    Since $g$ is isomorphic to $\Delta P_i$ in $G'_t$, $g$ is also isomorphic to $P$ in $G'_t$.
    Thus, $g \in R_{G'_t}(P)$.
    According to Definition~\ref{def:continuous_subgraph_enumeration}, $g \in \Delta R_t^+$.

    For any $g \in \Delta R_{t}^{i,-}$, the proof is similar.
\end{proof}

\begin{theorem}\label{thm:matching_result_must_be_isomorphic_subgraph}
    For any $g \in \Delta R_{t}^{+}$, $\exists i: g \in \Delta R_{t}^{i,+}$. For any $g \in \Delta R_{t}^{-}$, $\exists i: g \in \Delta R_{t}^{i,-}$.
\end{theorem}

\begin{proof}
    For any $g \in \Delta R_t^+$, $g$ is isomorphic to $P$ in $G'_t$.
    According to Definition~\ref{def:isomorphic_subgraph}, $|V(P)|=|V(g)|$ and $|E(P)|=|E(g)|$.
    There is one and only one match $f$ of $P$ in $g$ satisfying the partial order constraint.
    We number the edges of $g$ according to $P$.
    For the $k$-th edge of $P$ $e_k^P=(s_k,t_k)$, we number the edge $e_k^g=(f(s_k), f(t_k))$ of $g$ as $k$, correspondingly.
    As $g$ is a subgraph of $G'_t$, the edges of $g$ are also classified as the \textit{delta} edges and the \textit{unaltered} edges.
    Assume $g$ has $j$ delta edges $\{e_{x_1}^g, e_{x_2}^g, \dots, e_{x_j}^g\}$, where $x_1 < x_2 < \dots < x_j$.
    Now we proof that $g$ is isomorphic to $\Delta P_{x_j}$ in $G'_t$.
    For every edge $e_k^P=(s_k,t_k)$ of $P$ with $1 \leq k \leq |E(P)|$,
    \begin{enumerate}
        \item If $k < x_j$, $\tau_{x_j}(e_k^P) = \text{either}$ and $(f(s_k),f(t_k)) \in E(G'_t)$, because $f$ is a match of $P$ in $g$;
        \item If $k = x_j$, $\tau_{x_j}(e_k^P) = \text{delta}$ and $(f(s_k),f(t_k))$ is a delta edge, because $(f(s_k),f(t_k)) = e_{k}^g$ under the edge numbering and $e_{x_j}^g$ is a delta edge;
        \item If $k > x_j$, $\tau_{x_j}(e_k^P) = \text{unaltered}$ and $(f(s_k), f(t_k))$ is an unaltered edge, because $(f(s_k),f(t_k)) = e_{k}^g$ under the edge numbering and $e_{k}^g$ with $k > x_j$ are  unaltered edges.
    \end{enumerate}

    According to Definition~\ref{def:incremental_match}, $f$ is an incremental match of $\Delta P_{x_j}$ in $G'_t$.
    According to Definition~\ref{def:isomorphic_subgraph_of_incremental_pattern_graph}, $g$ is isomorphic to $\Delta P_{x_j}$.
    Thus, $g \in \Delta R_t^{x_j,+}$.

    For any $g \in \Delta R_t^-$, the proof is similar.
\end{proof}

Based on the two theorems, we can get Theorem~\ref{thm:matching_results_equivalence}.
It indicates that we can get the matching results $\Delta R_{t}^+$ by combining $\Delta R_t^{i,+}$ ($1 \leq i \leq m$) and get $\Delta R_t^-$ by combining $\Delta R_t^{i,-}$ ($1 \leq i \leq m$).
Theorem~\ref{thm:matching_result_do_not_overlap} further indicates that matching results of different incremental pattern graphs $\Delta R_t^{i,+}$ ($\Delta R_t^{i,-}$) do not overlap with each other.
We can simply combine $\Delta R_t^{i,+}$ ($\Delta R_t^{i,-}$) without de-duplicating.

\begin{theorem}\label{thm:matching_results_equivalence}
    $\Delta R_{t}^{+} = \bigcup_{1\leq i\leq m}\Delta R_{t}^{i,+}$ and $\Delta R_{t}^{-} = \bigcup_{1\leq i\leq m}\Delta R_{t}^{i,-}$, where $m=|E(P)|$.
\end{theorem}

\begin{proof}
    \newcommand{\UNION}{\bigcup_{1 \leq i \leq m}\Delta R_{t}^{i,+}}
    We first proof $\Delta R_{t}^+ \subseteq \UNION$.
    For any $g \in \Delta R_t^+$, according to Theorem~\ref{thm:matching_result_must_be_isomorphic_subgraph}, $\exists i: g \in \Delta R_t^{i,+}$.
    Thus, $g \in \UNION$.

    We then proof $\UNION \subseteq \Delta R_t^+$.
    For any $1 \leq i \leq m$ and any $g \in \Delta R_t^{i,+}$, according to Theorem~\ref{thm:isomorphic_subgraph_must_be_matching_result}, $g \in \Delta R_t^+$.
    Thus, $\Delta R_t^{i,+} \subseteq \Delta R_t^+$.
    $\UNION \subseteq \Delta R_t^+$.
    Given the above, $\Delta R_t^+ = \UNION$.

    The proof of $\Delta R_t^- = \bigcup_{1 \leq i \leq m} \Delta R_t^{i,-}$ is similar.
\end{proof}

\begin{theorem}\label{thm:matching_result_do_not_overlap}
    $\forall a\forall b$ with $1\leq a<b\leq|E(P)|$: $\Delta R_{t}^{a,+}\cap\Delta R_{t}^{b,+}=\emptyset$, $\Delta R_{t}^{a,-}\cap\Delta R_{t}^{b,-}=\emptyset$.
\end{theorem}

\begin{proof}
    We proof the theorem by contradiction.
    Assume $\exists a \exists b$ with $1 \leq a < b \leq |E(P)|$: $\Delta R_t^{a,+} \cap \Delta R_t^{b,+} \neq \emptyset$.
    For any $g \in \Delta R_t^{a,+} \cap \Delta R_t^{b,+}$, $g$ is a subgraph of $G'_t$.
    $g$ is isomorphic to $\Delta P_a$ and $\Delta P_b$ with the incremental match $f$ and $f'$, respectively.
    Since $g$ is also isomorphic to $P$, $g$ corresponds to one and only one match $f''$ satisfying the partial order constraints.
    According to Definition~\ref{def:incremental_match}, $f$ and $f'$ are also matches of $P$ in $G'_t$ satisfying the partial order constraints.
    Thus, $f=f'=f''$.
    Consider the $b$-th edge $e_b^P=(s_b,t_b)$ of $P$.
    Since $a < b$, $\tau_a(e_b^P)=\text{unaltered}$ and $\tau_b(e_b^P)=\text{delta}$.
    According to Definition~\ref{def:incremental_match}, $(f(s_b),f(t_b))$ is an unaltered edge, while $(f'(s_b), f'(t_b))$ is a delta edge.
    Since $f=f'$, $(f(s_b),f(t_b)) = (f'(s_b),f'(t_b))$.
    $(f(s_b),f(t_b))$ cannot be an unaltered edge and a delta edge at the same time.
    The assumption is invalid.
    Thus, $\forall a \forall b$ with $1 \leq a < b \leq |E(P)|$: $\Delta R_t^{a,+} \cap \Delta R_t^{b,+} = \emptyset$.

    The proof of $\forall a \forall b$ with $1 \leq a < b \leq |E(P)|$: $\Delta R_t^{a,-} \cap \Delta R_t^{b,-} = \emptyset$ is similar.

\end{proof}

With Theorem~\ref{thm:matching_results_equivalence} and Theorem~\ref{thm:matching_result_do_not_overlap}, we can turn the continuous subgraph enumeration into a series of subgraph enumeration on the snapshots.
For each time step $t$, we enumerate all isomorphic subgraphs of all incremental pattern graphs in the snapshot $G'_t$ and $G'_{t-1}$.
The isomorphic subgraphs form $\Delta R_t^+$ and $\Delta R_t^-$ at time step $t$ without duplication or omission.

Enumerating isomorphic subgraphs of incremental pattern graphs can be further converted into enumerating incremental matches.
There is a \emph{bijective} mapping between an isomorphic subgraph $g$
of $\Delta P_{i}$ in $G'_{t}$ ($G'_{t-1}$) and an incremental match $f$ of $\Delta P_{i}$ in $G'_{t}$ ($G'_{t-1}$).
On the one hand, a subgraph $g$ isomorphic to $\Delta P_{i}$ in $G'_{t}$ ($G'_{t-1}$) corresponds to only one incremental match $f$.
The incremental match $f$ is also a match of $P$ in $G'_t$ satisfying the partial order constraints.
According to the symmetry breaking technique, $f$ is unique.
On the other hand, an incremental match $f$ naturally corresponds to a subgraph $g$ in $G'_{t}$ ($G'_{t-1}$).
Based on the bijective mapping, we solve the continuous subgraph enumeration problem by enumerating incremental matches at each time step.

\subsection{Framework Overview}

We propose the \emph{Streaming-BENU} framework (S-BENU, for short) to enumerate incremental matches in snapshots efficiently.
The input of S-BENU contains the pattern graph $P$, the initial data graph $G'_0$, and the batch update $\Delta o_t$ at each time step $t$.
S-BENU outputs $\Delta R_{t}^{+}$ and $\Delta R_{t}^{-}$ at each time step.

Some existing continuous subgraph enumeration methods \cite{SJ-Tree}  \cite{TurboFlux} \cite{Timing} \cite{D-IDS} maintain the (partial) matching results of each time step in memory or on disk.
They use the matching results of the time step $t$ to compute the matching results of the time step $t+1$, avoiding re-computing some intermediate results.
However, maintaining matching results is only feasible when the pattern graph is highly selective \cite{SJ-Tree} and the size of matching results is not big.
If the data graph is big, the size of the matching results will become considerable as indicated by Table~\ref{tab:statistics_of_typical_pattern_graphs}.
Moreover, users may monitor multiple pattern graphs simultaneously, multiplying the storage cost.
In S-BENU, we choose not to maintain any matching result.
Instead, S-BENU only stores the data graph in a distributed database and shares the data graph among all pattern graphs.
S-BENU can store matching results in a distributed file system when it is needed.

Algorithm~\ref{alg:SBENU_framework} presents the pseudo-code of the S-BENU framework.
S-BENU consists of two phases.

\begin{algorithm}[!t]
    \small

    \textbf{Input:} Pattern graph $P$, Data graph $G'$, Distributed database $DB$. \\
    \textbf{Output:} $\Delta{R}_{t}^+$, $\Delta{R}_{t}^{-}$ for each time step $t$.

    \begin{algorithmic}[1]
        \State Store the initial graph of $G'$ into $DB$;
        \Comment{\emph{Phase I: Initialization}}
        \State $\mathbb{E}$ $\leftarrow$ \Call{GenerateBestExecutionPlans}{$P$};
        \State Broadcast $P$ and $\mathbb{E}$ to all worker nodes;
        \State $t \leftarrow 0$;
        \Loop
        \Comment{\emph{Phase II: Continuous enumeration}}
        \State{$t \leftarrow t + 1$}; \Comment{$t$ is the current time step}
        \State{Get the batch update $\Delta o_t$ of $G'$};
        \State Convert $\Delta o_t$ into delta adjacency sets;
        \State Store delta adjacency sets into $DB$;
        \ForAll{$start \in \{v|\Delta\Gamma_t^{\text{out}}(v) \neq \emptyset\}$} \emph{in parallel}
        \ForAll{$E_i \in \mathbb{E}$} \Comment{Execute $E_i$}
        \State $f$ $\leftarrow$ an empty mapping from $V(P)$ to $V(G')$;
        \State // \emph{Denote the $i$-th edge of $P$ as $e^P_i=(u_{s_i}, u_{t_i})$}
        \State $f(u_{s_i}) \leftarrow start$;
        \ForAll{$(op, v_y) \in Filter(\Delta\Gamma_t^{\text{out}}(f(u_{s_i}))$}
            \State $f(u_{t_i}) \leftarrow v_y$;
            \If{$op = +$} \Comment{Enumerate matches in $\Delta R_{t}^{i,+}$}
            \State Match remaining vertices in $f$ in $G'_t$;
            \Else \Comment{Enumerate matches in $\Delta R_{t}^{i,-}$}
            \State Match remaining vertices in $f$ in $G'_{t-1}$;
            \EndIf
            \EndFor
            \EndFor
            \EndFor
            \State Merge adjacency sets in $DB$ with delta adjacency sets;
        \EndLoop
    \end{algorithmic}

    \caption{Streaming-BENU Framework\label{alg:SBENU_framework}}
\end{algorithm}

The initialization phase is conducted once for $G'$.
S-BENU stores the initial graph of $G'$ into a distributed key-value database $DB$ in parallel.
For a vertex $v$, the key is its ID and the value is a tuple of its incoming and outgoing adjacency sets.
We will elaborate on the structure of the tuple later in Section~\ref{sec:data_graph_storage}.
The edges in the pattern graph $P$ are numbered.
S-BENU generates the best incremental execution plans $\mathbb{E}=\{E_1, E_2, \dots, E_n\}$ for every incremental pattern graph $\Delta P_i$ ($1 \leq i \leq |E(P)|$).
$P$ and $\mathbb{E}$ are then broadcasted.

The continuous enumeration phase is conducted repeatedly for every time step $t$.
It consists of three sub-phases: pre-processing (lines 7-9), enumeration (lines 10-20) and post-processing (line 21).

In pre-processing, S-BENU gets $\Delta o_{t}$ from an external data source like a message queue or a file.
S-BENU converts $\Delta o_{t}$ into \emph{delta adjacency sets} in parallel.
For a vertex $v$, its \emph{delta} adjacency sets are $\Delta\Gamma_{t}^{\text{in}}(v)=\{(op,w)|(op,w,v)\in\Delta o_{t}\}$ and $\Delta\Gamma_{t}^{\text{out}}(v)=\{(op,w)|(op,v,w)\in\Delta o_{t}\}$.
S-BENU only generates delta adjacency sets for the vertices appearing in $\Delta o_t$.

In enumeration, S-BENU generates a local search task for every vertex $start$ that has a non-empty delta outgoing adjacency set.
The $start$ is the starting vertex of the local search task.
S-BENU executes local search tasks in parallel in a distributed computing platform (line 10).
In every task, S-BENU executes incremental execution plans $E_i$ one by one (line 11).
$E_i$ searches incremental matches of $\Delta P_i$ in both $G'_t$ and $G'_{t-1}$ (lines 12 - 20).
We elaborate on it later in Section~\ref{sec:incremental_execution_plan}.

In post-processing, S-BENU fetches adjacency sets $\Gamma_{G'_{t-1}}^{\text{in/out}}(v)$ of the vertices appearing in $\Delta o_t$ from $DB$, merges them with $\Delta\Gamma_t^{\text{in/out}}(v)$, and stores $\Gamma_{G'_t}^{\text{in/out}(v)}$ back to $DB$ in parallel (line 21).

\subsection{Incremental Execution Plan\label{sec:incremental_execution_plan}}

The incremental execution plan is the core of S-BENU.
The $i$-th incremental execution plan $E_i$ gives out the detailed steps to enumerate incremental matches of $\Delta P_i$ in both $G'_t$ and $G'_{t-1}$.

\subsubsection{Abstract Plan}

$E_i$ implements the three core functions in Algorithm~\ref{alg:backtracking_based_framework}  as the following.

\noun{First(Next)PatternVertexToMatch}:
$E_i$ is bound with a static matching order $O_i$.
Similar to B-BENU, the function returns the first pattern vertex in $O_i$ that is unmapped in the partial incremental match $f$.
S-BENU has an extra constraint on $O_i$.
Suppose the $i$-th edge of $P$ is $e^P_i = (u_{s_i}, u_{t_i})$.
The first two pattern vertices in $O_i$ must be $u_{s_i}$ and $u_{t_i}$.
S-BENU relies on them to determine which snapshot $f$ belongs to.
As $\tau_i(e^P_i)=\text{delta}$ in $\Delta P_i$, $f$ must map $e_i^P$ to a \emph{delta} edge $e^{G'}$.
If $e^{G'}$ is an inserting edge ($op=+$), $e^{G'} \in E(G'_t)$ and $f$ should be an incremental match of $\Delta P_i$ in $G'_t$.
Otherwise, $e^{G'}$ is a deleting edge ($op=-$) that belongs to $G'_{t-1}$.
In this case, $f$ is an incremental match in $G'_{t-1}$.

\noun{RefineCandidates}:
Similar to B-BENU, S-BENU also calculates candidate set $C_{j}$ for $u_{j}$ by intersecting adjacency sets of already mapped vertices.
Since edges of $\Delta P_i$ have three types (either, delta, unaltered) and two directions(in, out), there are six kinds of adjacency sets.
In the snapshot $G'_?$ ($?$ can be $t$ or $t-1$), the incoming adjacency sets of $v$ are
$\Gamma_{G'_?}^{\text{either,in}}(v)=\{v_x|(v_x,v) \in E(G'_?)\}$,
$\Gamma_{G'_?}^{\text{unaltered,in}}(v) = \{v_x| (v_x,v) \in E(G'_?), (v_x, v) \text{ is an unaltered edge}\}$, and
$\Gamma_{G'_?}^{\text{delta,in}}(v) = \{v_x| (v_x,v) \in E(G'_?), (v_x,v) \text{ is a delta edge}\}$.
The definitions of outgoing adjacency sets are similar.
$?$ is determined by $op$.
If $op=+$, $f$ is in $G'_t$ and $?$ is $t$.
Otherwise, $?$ is $t-1$.

We use $\mathcal{N}_i^\text{in}(u_j)$/$\mathcal{N}_i^\text{out}(u_j)$ to denote the set of $u_j$'s incoming/outgoing neighbors in $P$ that are before $u_j$ in $O_i$, i.e. $\mathcal{N}_i^\text{in/out}=\{u_x|u_x \in \Gamma_P^\text{in/out}(u_j), u_x \text{ is before } u_x \text{ in } O_i\}$.
The pattern vertices in $\mathcal{N}_i^\text{in}(u_j)$ and $\mathcal{N}_i^\text{out}(u_j)$ are already mapped in $f$ when we calculate $C_j$.
If we want to map $u_j$ to $v$ in $f$, $v$ should satisfy two conditions:
1) for any mapped incoming neighbor $u_x \in \mathcal{N}_i^\text{in}(u_j)$, $v \in \Gamma_{G'_?}^{\tau_i((u_x,u_j)), \text{out}}(f_x)$;
2) for any mapped outgoing neighbor $u_x \in \mathcal{N}_i^\text{out}(u_j)$, $v \in \Gamma_{G'_?}^{\tau_i((u_j, u_x)), \text{in}}(f_x)$.
Based on the conditions, the candidate set $C_j$ is
\begin{align*}
    C_j = & (\cap_{u_x \in \mathcal{N}_i^\text{in}(u_j)}{\Gamma_{G'_?}^{\tau_i((u_x,u_j)),\text{out}}(f_x)})        \\
          & \cap  (\cap_{u_x \in \mathcal{N}_i^\text{out}(u_j)}{\Gamma_{G'_?}^{\tau_i((u_j,u_x)),\text{in}}(f_x)}).
\end{align*}
Mapping $u_j$ to any vertex outside $C_j$ will violate Definition~\ref{def:incremental_match}.
$C_j$ is further filtered to ensure that data vertices in it do not violate injective conditions and partial order constraints.

\begin{example}
    Take the toy case in Fig.~\ref{fig:demo_case_of_continuous_subgraph_enumeration} with $\Delta P_2$ and $t=2$ as the example.
    Since $e_2^P=(u_1,u_3)$, $O_2$ must be $u_1, u_3,u_2$.
    Suppose $u_2$ is unmapped in the partial match $f=(v_1, ?, v_4)$.
    $f$ maps $e_2^P$ to an inserting edge $(v_1,v_4)$ with $op=+$.
    Thus, we use adjacency sets from $G'_2$ to calculate $C_2$.
    $\mathcal{N}_2^\text{in}(u_2)=\{u_1\}$ and $\mathcal{N}_2^\text{out}(u_2)=\{u_3\}$.
    $C_2 = \Gamma_{G'_2}^{\text{either},\text{out}}(f_1) \cap \Gamma_{G'_2}^{\text{unaltered},\text{in}}(f_3)
        = \Gamma_{G'_2}^{\text{either,out}}(v_1) \cap \Gamma_{G'_2}^{\text{unaltered,in}}(v_4)
        = \{v_4,v_5,v_6,v_7,v_8\} \cap \{v_3,v_6\} = \{v_6\}$.
    $f=(v_1,v_6,v_4)$ is an incremental match of $\Delta P_2$ in $G'_2$.
\end{example}

Integrating the core functions, we can get an abstract incremental execution plan.
Fig.~\ref{fig:abstract_incremental_execution_plan} shows the abstract plan of $\Delta P_2$ in Fig.~\ref{fig:demo_case_of_continuous_subgraph_enumeration}.
Line 1 to line 3 in Fig.~\ref{fig:abstract_incremental_execution_plan} correspond to line 12 to line 16 in Algorithm~\ref{alg:SBENU_framework}.
$op$ of the delta edge is retrieved simultaneously during the mapping of the second pattern vertex $u_3$ in $O_2$.
Line 4 to line 6 in the abstract plan correspond to line 17 to line 20 in Algorithm~\ref{alg:SBENU_framework}.

\begin{figure}[!t]
    \centering
    \subfloat[Abstract\label{fig:abstract_incremental_execution_plan}]{\includegraphics[width=0.4\columnwidth]{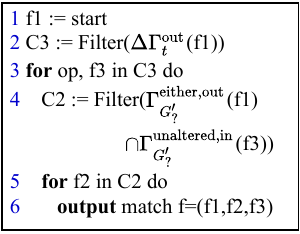}}
    ~
    \subfloat[Raw\label{fig:raw_incremental_execution_plan}]{\includegraphics[width=0.45\columnwidth]{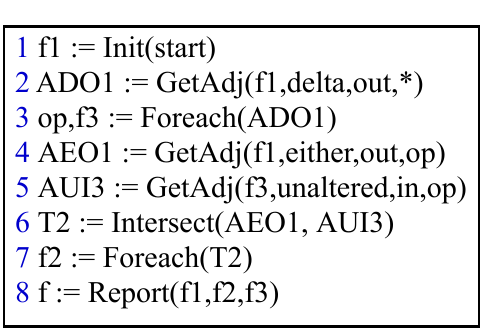}}
    \caption{Incremental execution plan for $\Delta P_{2}$ with $O_{2}:u_{1,}u_{3},u_{2}$.\label{fig:demo_incremental_execution_plan}}
\end{figure}

\subsubsection{Concrete Plan}

We materialize the abstract incremental execution plan into a concrete one with seven kinds of execution instructions listed in Table~\ref{tab:types_of_execution_instructions}.
S-BENU inherits the INI, INT, ENU, and RES instructions from B-BENU, but it modifies the DBQ instruction and adds the Delta-ENU and INS instructions.

The modified DBQ instruction adds three extra parameters as $X := GetAdj (f_i,type,direction,op)$.
It fetches the adjacency set $\Gamma_{?}^{\text{type,direction}}(f_i)$ from the distributed database, where $type$ can be either/delta/unaltered, and $direction$ can be in/out.
If $op=+$ (or $-$), the instruction gets the adjacency set $\Gamma_{G'_{t(\text{or } t-1)}}^\text{type, direction}(f_i)$.
If $type=\text{delta}$ and $op=*$, the instruction gets the delta adjacency set $\Delta\Gamma_t^\text{direction}(f_i)$ of the current time step $t$.
Vertices in the delta adjacency set are attached with flags ($+$ or $-$).
To reference adjacency sets consistently in the plan, the target variable $X$ conforms to a special naming convention.
The name of $X$ consists of three letters and a subscript.
The first letter is always \emph{A}, representing \emph{a}djacency sets.
The second letter can be \emph{E/D/U}, depending on $type$ (Either/Delta/Unaltered).
The third letter can be \emph{I/O}, depending on $direction$ (In/Out).
The subscript is the index of the operand $f_i$ $i$.
For example, $ADO_3 := GetAdj(f_3, delta, out, *)$ fetches $\Delta\Gamma_t^\text{out}(f_3)$ from the database.
Taking the case in Fig.~\ref{fig:demo_case_of_continuous_subgraph_enumeration} with $t=2$ and $f_3=v_1$ as the example, $\Delta\Gamma_2^\text{out}(f_3)=\{(-,v_2), (-,v_3), (+,v_4)\}$.

S-BENU extends the ENU instruction to the Delta-ENU instruction $op, f_i := Foreach(X)$.
The instruction requires $X$ to be a (filtered) delta adjacency set.
The instruction retrieves $op$ and $f_i$ simultaneously while it iterates $X$.

Given $\Delta P_i$, to materialize its abstract plan, S-BENU generates instructions for each pattern vertex following a given matching order $O_i: u_{k_1}, u_{k_2}, \dots, u_{k_n}$.
For the first two pattern vertices in $O_i$, S-BENU generates four instructions consequently to implement lines 12 to 16 of Algorithm~\ref{alg:SBENU_framework}:
\begin{enumerate}
    \item $f_{k_1} := Init(start)$.
    \item $ADO_{k_1} := GetAdj(f_{k_1}, delta, out, *)$.
    \item $C_{k_2} := Intersect(ADO_{k_1}) | [FCs]$.
    \item $op, f_{k_2} := Foreach(C_{k_2})$.
\end{enumerate}
S-BENU then generates DBQ instructions to fetch the \{either, unaltered\} $\times$ \{in, out\} adjacency sets of $u_{k_1}$ and $u_{k_2}$, in case some INT instructions of other vertices may use them.
S-BENU does not need to consider adjacency sets of the type delta, since there is only one delta edge in $\Delta P_i$.

S-BENU only checks the existence of $(f_{k_1}, f_{k_2})$ when it generates $C_{k_2}$.
If there is an edge $(u_{k_2}, u_{k_1}) \in E(P)$, S-BENU checks the existence of $(f_{k_2}, f_{k_1})$ in $G'_{?}$ with an extra INS instruction $InSetTest(f_{k_1}, A?O_{k_2})$, where $?$ depends on $\tau_i((u_{k_2}, u_{k_1}))$.
If $f_{k_1}$ is not in the outgoing adjacency set of $f_{k_2}$, S-BENU backtracks and tries to map $u_{k_2}$ to another data vertex in $C_{k_2} $.

For each of the remaining vertices $u_{k_j}$ in $O_i$, S-BENU generates instructions for it in a way similar to B-BENU.
$T_{k_j} := Intersect(...)$ and $C_{k_j} := Intersect(T_{k_j}) [|FCs]$ calculate the candidate set $C_{k_j}$ with related adjacency sets.
$f_{k_j} := Foreach(C_{k_j})$ maps $u_{k_j}$ to the candidate data vertices in $C_{k_j}$ one by one and enters the next level in the backtracking search.
S-BENU generates DBQ instructions to fetch the \{either, unaltered\} $\times$ \{in, out\} adjacency sets of $u_{k_j}$.

Finally, S-BENU adds a RES instruction to the execution plan.

After generating all instructions, S-BENU first removes useless DBQ instructions whose target variables are not used by any other INT/INS instruction.
Then S-BENU conducts the \emph{uni-operand elimination}.
It removes useless INT instructions without any filtering condition like $T_x := Intersect(X)$ and replaces $T_x$ with $X$ in other instructions.
After uni-operand elimination, S-BENU gets the raw incremental execution plan.
The raw plan is well-defined.
All variables are defined before used.

\begin{example}
    Fig.~\ref{fig:raw_incremental_execution_plan} materializes the abstract plan of $\Delta P_2$ in Fig.~\ref{fig:abstract_incremental_execution_plan}.
    In Fig.~\ref{fig:raw_incremental_execution_plan}, the instruction 1 to 5 are generated for $u_1$ and $u_3$.
    Instructions 6 to 7 are generated for $u_2$.
    Some DBQ/INT instructions are useless, like the ones related to $AEI_1$/$AEO_2$/$C_3$/$C_2$.
    They are removed from the raw plan.
\end{example}

\subsection{Best Execution Plan Generation}

S-BENU optimizes the raw execution plan with the \emph{common subexpression elimination} and the \emph{instruction reordering} optimizations as proposed in Section~\ref{subsec:execution_plan_optimization}.
We do not adopt the triangle caching optimization in S-BENU because edges are typed and directed in incremental pattern graphs, making it hard to re-use the enumerated triangles.

S-BENU modifies Algorithm~\ref{alg:best_execution_plan_generation} to generate the best execution plan for each incremental pattern graph $\Delta P_i$.

Suppose the $i$-th edge of $P$ is $e_i^P=(u_{s_i},u_{t_i})$.
The first two vertices in candidate matching orders are fixed as $u_{s_i}$ and $u_{t_i}$.

The dual condition in the dual pruning technique is stricter.
In $\Delta P_i$, the neighborhood of $u_x$ is \emph{contained} by the neighborhood of $u_y$ if
\begin{enumerate}
    \item For every $e=(u_z, u_x) \in E(P)$ with $u_z \neq u_y$, $e'=(u_z,u_y) \in E(P)$ and $\tau_i(e)=\tau_i(e')$;
    \item For every $e=(u_x,u_z) \in E(P)$ with $u_z \neq u_y$, $e'=(u_y,u_z) \in E(P)$ and $\tau_i(e)=\tau_i(e')$.
\end{enumerate}
$u_x$ and $u_y$ is \emph{syntactic equivalent} if and only if  the neighborhood of $u_x$ is contained by the neighborhood of $u_y$ and vice versa.

When S-BENU estimates the number of matching results of partial pattern graphs, S-BENU treats them as undirected graphs and uses the model in \cite{SEED} to estimate.
Though the model is targeted for undirected graphs, we find it good enough in practice to distinguish good matching orders from bad ones on directed graphs.
Proposing a more accurate estimation model for incremental pattern graphs is one of our future work.

\section{Efficient Implementation}

Fig.~\ref{fig:BENU_architecture} shows the implementation architecture of B-BENU and S-BENU.
During the initialization phase, the data graph $G$ is storesd into a distributed key-value database like HBase.
The update stream $\Delta o_t$ in S-BENU is got from a file or a message queue.
B-BENU/S-BENU generates local search tasks from $G$/$\Delta o_t$ in parallel, respectively.
The tasks are executed in a distributed computing platform like Hadoop and Spark.
Building upon a distributed computing platform and a distributed key-value database, B-BENU and S-BENU naturally support fault tolerance.
We further propose several implementation techniques to increase efficiency.

\begin{figure}[!t]
    \centering
    \includegraphics[width=0.98\columnwidth]{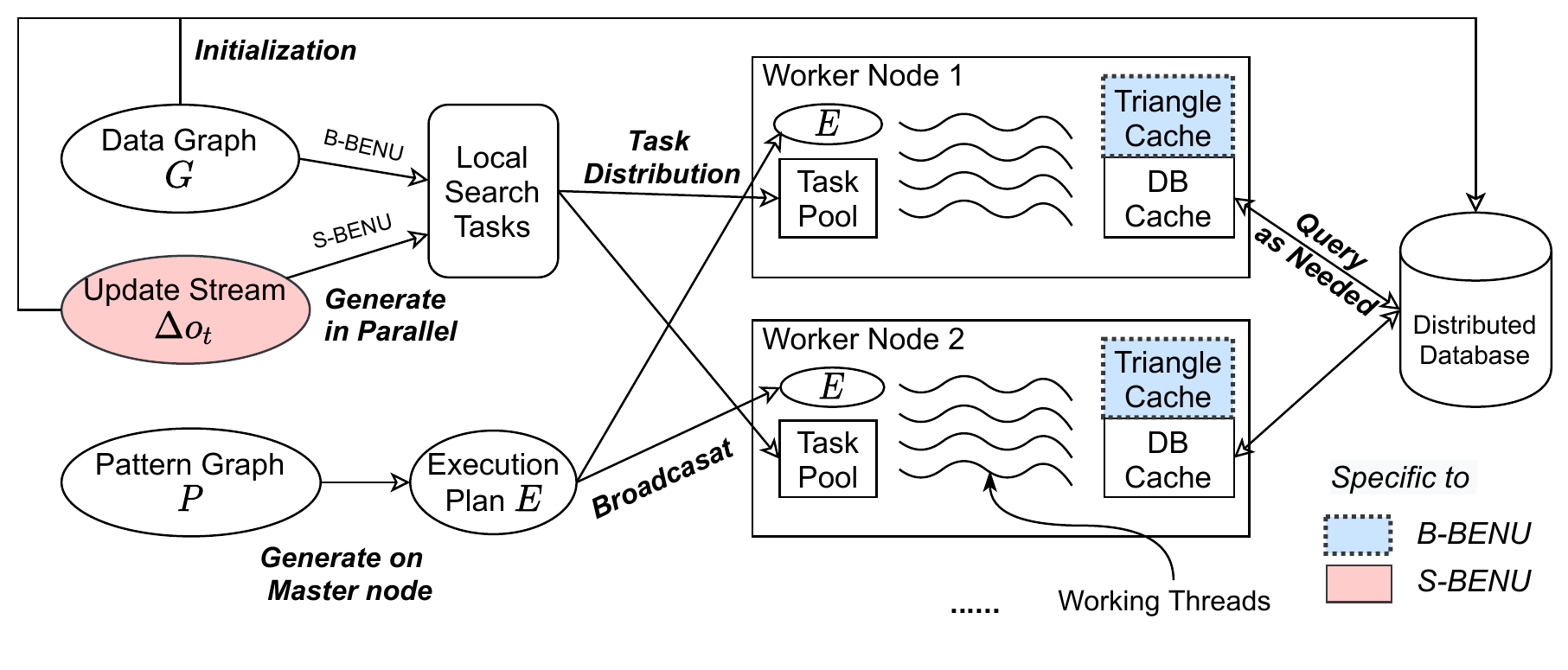}
    \caption{Architecture of B-BENU and S-BENU.\label{fig:BENU_architecture}}
\end{figure}

\subsection{Local Database Cache}

Inside a local search task, a queried adjacency set tends to be queried again soon by the same task.
For example, in the backtracking search trees illustrated in Fig.~\ref{fig:backtracking_search_tree}, the adjacency set of $v_4$ is queried repeatedly in different search branches in the local search task 1.
This kind of locality comes from the backtracking nature of the execution plan.
All vertices that a local search task visits are in a local neighborhood around the starting vertex of the task.
The size of the local neighborhood is bounded by the radius of the pattern graph that is usually small.
When a local search task queries many adjacency sets during the backtracking search, there are some repeated queries, bringing \emph{intra-task locality}.

Some adjacency sets are queried by many different local search tasks.
For example, in Fig.~\ref{fig:backtracking_search_tree}, the adjacency set of $v_4$ is queried in both Task 1 and Task 2.
This kind of \emph{inter-task locality} comes from the overlaps between local neighborhoods visited by different tasks.
Some data vertices, especially high-degree ones, are included in many local neighborhoods.
Their adjacency sets are frequently queried by different tasks.

To take advantage of both kinds of locality, we set up an in-memory \emph{database cache} (DB cache) with configurable capacity in each worker node as shown in Fig.~\ref{fig:BENU_architecture}.
The DB cache stores adjacency sets fetched from the distributed database.
The cache captures the intra-task locality by using advanced replacement policies like LRU.
It captures the inter-task locality by being shared among all working threads.
The cache provides a flexible mechanism to \emph{trade memory for the reduction in communication}.

\paragraph*{Complexity Analysis}
With the cache technique, the communication cost of an execution plan $E$ (i.e. the number of conducted database queries) is also related to the cache capacity $C$.
To analyze its upper bound, we first define several concepts.
The $r$-hop neighborhood ($r \geq 0$) of a vertex $v$ in a graph $g$ is defined as $\gamma_g^r(v)$=\{$w \in V(g) |$ $w$  is at most $r$ hops away from $v$\}.
The size of $\gamma_g^r(v)$ is $S_g^r(v)=\sum_{w \in \gamma_g^r(v)}{d_g(w)}$.
For a data graph $G$, $H_G^r = \max_{v \in V(G)}{S_G^r(v)}$ is the size of the largest $r$-hop neighborhood in $G$.
As for the cache, we assume there is $w$ working threads per machine and there exists $R$ that  $C \geq w H_G^{R}$, i.e. the cache can store the $R$-hop neighborhood of any data vertex for every working thread.
As for the execution plan, we assume its matching order is $O: u_{k_1}, u_{k_2}, \dots, u_{k_n}$.
The first $\alpha$ vertices in $O$ can cover every edge in $P$.
Thus, matching $u_{k_{\alpha+1}},\dots,u_{k_n}$ does not need to query any adjacency set.
Among the first $\alpha$ vertices, there must exist $r'$ ($0\leq r' \leq R$)  and $\beta$ ($1 \leq \beta \leq \alpha)$ that the $r'$-hop neighborhood of $u_{k_\beta}$ $\gamma_P^{r'}(u_{k_{\beta}})$ contains $\{u_{k_{\beta}}, u_{k_{\beta+1}}, \dots, u_{k_{\alpha}}\}$.
Then, we can split $O$ into three sections: $O: u_{k_1},\dots,u_{k_\beta},\dots,u_{k_\alpha},\dots,u_{k_n}$.

The total communication cost of matching $f_{k_1} $ to $ f_{k_{\beta}}$ is $O(\sum_{i=1}^{\beta}|{R}_G(P_i)|)$ where $|{R}_G(P_i)|$ is the number of matches of the partial pattern graph $P_i$ in $G$.
If $f_{k_1}$ to $f_{k_\beta}$ is fixed in $f$, the number of conducted database queries during matching $f_{k_{\beta+1}}$ to $f_{k_{\alpha}}$ is at most $\max_{v \in V(G)}|\gamma_G^{r'}(v)|$, because the cache can store all the adjacency sets in $\gamma_G^{r'}(f_{k_\beta})$.
The total communication cost of matching $f_{k_{\beta+1}}$ to $f_{k_{\alpha}}$ for all the partial matches is  ${O}(|{R}_G(P_\beta)|\max_{v \in V(G)}{|\gamma_G^{r'}(v)|})$.
Matching the remaining vertices $f_{k_{\alpha+1}}$ to $f_{k_n}$ does not query any adjacency set.
Therefore, the communication upper bound is ${O}(\sum_{i=1}^{\beta}|{R}_G(P_i)| + |{R}_G(P_\beta)|\max_{v \in V(G)}{|\gamma_G^{r'}(v)|})$.

If ${C}$ is bigger than the data graph, a tighter upper bound is ${O}(p|V(G)|)$ where $p$ is the number of worker machines.
In this case, the complexity is independent of the pattern graph.


\subsection{Data Graph Storage\label{sec:data_graph_storage}}

B-BENU and S-BENU store adjacency sets of the data graph with key-value pairs.
For B-BENU, keys are vertex IDs and values are their adjacency sets.

For S-BENU, we notice that it only uses the latest two snapshots $G'_t$ and $G'_{t-1}$ during the execution.
Thus, we only maintain two snapshots  in the database.
For a vertex $v$, the key is its ID, and the value is a quad $(?,?,?,?)$.
The value has two forms.
The first form is $(\Gamma_{G'_t}^\text{in}(v), \Gamma_{G'_t}^\text{out}(v), \emptyset, \emptyset)$.
It is used in Line 1 and Line 21 of Algorithm~\ref{alg:SBENU_framework} to store the current snapshot $G'_t$.
The second form is $(\Gamma_{G'_{t-1}}^\text{in}(v), \Gamma_{G'_{t-1}}^\text{out}(v), \Delta\Gamma_t^\text{in}(v), \Delta\Gamma_t^\text{out}(v))$.
It is used in Line 9 of Algorithm~\ref{alg:SBENU_framework} to store the delta adjacency sets along with the previous snapshot $G'_{t-1}$.
With this form, we can retrieve the adjacency sets of both $G'_{t-1}$ and $G'_t$.
The two-form design guarantees that we only need to update the vertices appearing in $\Delta o_t$ in Line 9 and Line 21.
For a vertex $v$ not appearing in $\Delta o_t$, its value is $(\Gamma_{G'_{t-1}}^\text{in}(v), \Gamma_{G'_{t-1}}^\text{out}(v), \emptyset, \emptyset)$ before Line 9.
Since $\Delta\Gamma_t^\text{in/out}(v)=\emptyset$ and $\Gamma_{G'_{t-1}}^\text{in/out}(v)=\Gamma_{G'_t}^\text{in/out}(v)$, we do not need to modify its value in either Line 9 or Line 21.
As $|\Delta o_t| \ll |E(G'_t)|$, only a fraction of vertices appear in $\Delta o_t$.
The two-form design cuts much costs of updating the database.

In the local database cache, B-BENU uses the same key-value format as in the database, but S-BENU uses a different format.
For S-BENU, keys are still vertex IDs, but values are $(T, \Gamma_{G'_{T-1}}^\text{in}(v), \Gamma_{G'_{T-1}}^\text{out}(v),\Gamma_{G'_T}^\text{in}(v), \Gamma_{G'_T}^\text{out}(v))$.
$T$ is the time step of the key-value pair.
Vertices in all adjacency sets are attached with flags, indicating whether the corresponding edge is a \emph{delta} edge.
When a DBQ instruction $X := GetAdj(f_i, type, direction,op)$ is conducted, the cache hits if the key $f_i$ is contained in the cache and $T$ is equivalent to the current time step $t$.
If the cache hits, S-BENU retrieves the corresponding adjacency set based on \textit{op} and \textit{direction}, and S-BENU filters it with \textit{type}.
If the cache misses, S-BENU queries the distributed database for the quad $(\Gamma_{G'_{t-1}}^\text{in}(v), \Gamma_{G'_{t-1}}^\text{out}(v), \Delta\Gamma_t^\text{in}(v), \Delta\Gamma_t^\text{out}(v))$ and constructs the value part from it.
The format in the cache trades space for time, because the cache is much more frequently accessed than the database.
If we use the same format as in the database, we have to merge $\Delta\Gamma_t^\text{in/out}(v)$ with $\Gamma_{G'_{t-1}}^\text{in/out}(v)$ to get $\Gamma_{G'_t}^\text{?,in/out}(v)$.
Merging two adjacency sets is more expensive than filtering an adjacency set with flags.

\subsection{Task Splitting}

The computation and communication costs of a local search task are positively correlated with the degree of the starting vertex.
Unfortunately, real-world graphs often follow the power-law degree distribution, causing workloads of local search tasks skewed.
We propose the \emph{task splitting} technique to split heavy tasks into smaller subtasks to balance the workloads.
Suppose $u_{k_1}$ and $u_{k_2}$ are the first and second pattern vertex in the matching order and $C_{k_2}$ is the candidate set of $u_{k_2}$.

In B-BENU, $u_{k_1}$ is mapped to the starting vertex $v$ of the local search task.
If $u_{k_1}$ and $u_{k_2}$ are adjacent in $P$, $C_{k_2}$ is the filtered adjacency set of the starting vertex.
If the degree $d_G(v)$ is bigger than a given threshold $\theta$, we split $\Gamma(v)$ into $\lceil\frac{|\Gamma(v)|}{\theta}\rceil$ non-overlapping equal-sized subsets.
We generate a subtask for each subset and use the subset as $C_{k_2}$ in the subtask.
If $u_{k_1}$ and $u_{k_2}$ are not adjacent, $C_{k_2}$ is the filtered $V(G)$, and we generate $\lceil\frac{|V(G)|}{\theta}\rceil$ subtasks in this case.

In S-BENU, $u_{k_1}$ is also mapped to the starting vertex $v$ of the local search task, but $C_{k_2}$ is the filtered $\Delta\Gamma_t^\text{out}(v)$.
If $|\Delta\Gamma_t^\text{out}(v)| \geq \theta$, we split $\Delta\Gamma_t^\text{out}(v)$  into $\lceil\frac{|\Delta\Gamma_t^\text{out}(v)|}{\theta}\rceil$ non-overlapping equal-sized subsets and generate a subtask for each subset.

\subsection{Implementation Sketch}

We implement B-BENU with a Hadoop MapReduce job.
The input to the job is the data graph stored as key-value pairs in HDFS.
In the map phase, B-BENU stores the data graph into HBase in parallel.
B-BENU generates and emits local search (sub)tasks simultaneously.
Keys are tasks and values are null.
Hadoop shuffles the tasks to reducers.
B-BENU runs a reducer on each worker machine.
In each reducer, B-BENU uses a thread pool to execute the received tasks concurrently.

Since S-BENU needs to process batch updates $\Delta o_t$ iteratively, we implement S-BENU with a long-running Spark job.
S-BENU loads the initial data graph from HDFS as a RDD and stores it into HBase in parallel by conducting \texttt{foreachPartition} operator.
In each time step, S-BENU loads $\Delta o_t$ as a RDD from an external data source like HDFS or a message queue.
S-BENU converts the $\Delta o_t$ RDD into the delta adjacency set RDD with the \texttt{flatMap} and \texttt{groupByKey} operators.
The delta adjacency set RDD is used to update HBase in parallel (Line 9 of Algorithm~\ref{alg:SBENU_framework}).
It is further converted into the local search task RDD.
Conducting \texttt{mapPartition} operator on it, S-BENU executes local search tasks in parallel with all executors.
The delta adjacency set RDD is used again to update HBase (Line 21 of Algorithm~\ref{alg:SBENU_framework}).

\section{Experiments}

We introduce the experimental setup in Section~\ref{sec:experimental_setup} and then evaluate the effectiveness of the proposed optimization techniques in Section~\ref{sec:evaluation_of_optimization_techniques}.
The performance of B-BENU and S-BENU is compared with the state-of-the-art in Section~\ref{sec:comparing_benu_with_state_of_the_art} and Section~\ref{sec:comparing_sbenu_with_state_of_the_art}.
We finally evaluate the machine scalability of B-BENU and S-BENU in Section~\ref{sec:scalability}.

\subsection{Experimental Setup\label{sec:experimental_setup}}

\textbf{Environment}. All the experiments were conducted in a cluster with 1 master + 16 workers connected via 1Gbps Ethernet.
Each machine was equipped with 12 cores, 50 Gbytes memory, and 2 Tbytes RAID0 HDD storage.
All Java programs were compiled with JDK 1.8 and run under CentOS 7.0.
We adopted Hadoop 2.7.2.
The distributed database was HBase 1.2.6.

\textbf{B-BENU}. B-BENU was implemented with Hadoop MapReduce.
B-BENU generated local search tasks in the map phase and executed the tasks in the reduce phase.
We used 16 reducers (one reducer per machine).
Each reducer ran the local search tasks with 24 working threads.
We allocated 40 Gbytes memory to each reducer (30 Gbytes for local database cache and 10 Gbytes for task execution).
The degree threshold $\theta$ of task splitting was 500.
Without otherwise mentioned, we used compressed execution plans in the experiments related to B-BENU.

\textbf{S-BENU}. S-BENU was implemented with Spark 2.2.0.
All phases were implemented with RDDs.
We used 16 executors (one executor per machine).
Each executor used 24 cores (i.e. working threads) to run tasks.
We allocated 40 Gbytes memory to each executor (30 Gbytes for local database cache and 10 Gbytes for task execution).
We turned off the task split technique by default.
Without otherwise mentioned, the execution time of S-BENU was the wall-clock time spent on the continuous enumeration phase of S-BENU, as the initialization phase was conducted once.

\textbf{Data Graphs}. For B-BENU, we used five real-world static data graphs in Table~\ref{tab:statistics_of_typical_pattern_graphs}.
They were also used by the previous work \cite{SEED} \cite{CBF}.
For S-BENU, we used a real-world dynamic data graph \emph{Wikipedia} (denoted \emph{wk}) \cite{mislove-2009-socialnetworksthesis} with 1.9M vertices and 40.0M edges.
We also used the LDBC-SNB Data Generator \cite{LDBC-SNB-Data-Generator} provided by the LDBC Graphalytics Benchmark \cite{LDBC_Graphalytics} to generate a synthetic dynamic social network.
The scale factor of the generator was graphalytics.1000.
We used the person-knows-person part as the dynamic data graph (denoted as \emph{ld}) with 11M vertices and 0.93B edges.
We generated batch updates of dynamic graphs based on the creation time of edges.

\begin{figure}[!t]
    \centering
    {\includegraphics[height=1.2cm]{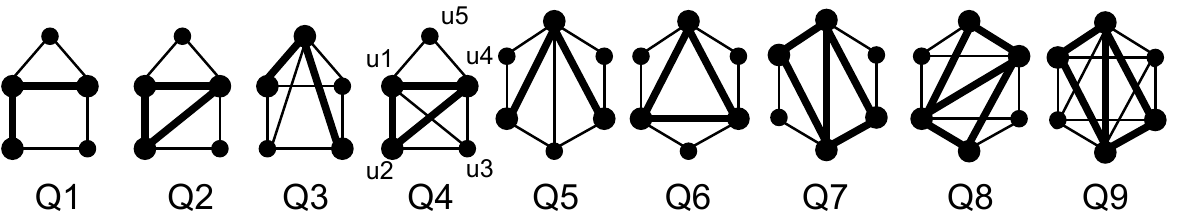}}  {\includegraphics[height=1.4cm]{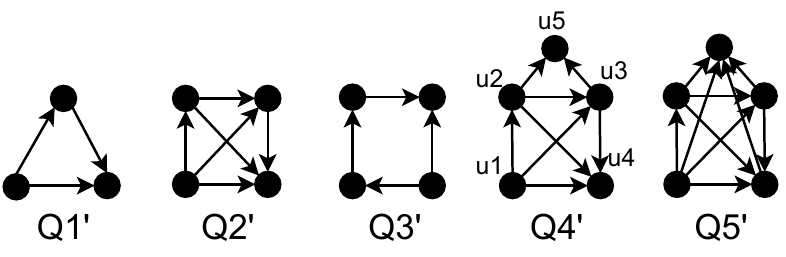}}
    \caption{Pattern graphs.\label{fig:pattern_graphs}}
\end{figure}

\textbf{Pattern Graphs} (Fig.~\ref{fig:pattern_graphs}). For B-BENU, we used Q1 to Q9.
Q1 to Q5 come from \cite{CBF}.
To evaluate the performance on tough tasks, we further used Q6 to Q9.
The vertex covers used in the VCBC compression are illustrated with big dots.
For S-BENU, we used Q1' to Q5' coming from \cite{BiGJoin}.

\subsection{Evaluation of Optimization Techniques\label{sec:evaluation_of_optimization_techniques}}

We evaluated the effectiveness of the proposed techniques in both B-BENU and S-BENU.

\textbf{Exp-1: Best Execution Plan Search}.
We evaluated the efficiency of Algorithm~\ref{alg:best_execution_plan_generation} with random connected graphs.
We generated 1000 Erdos-Renyi random pattern graphs for each number of vertices $n$.
For every pattern graph, we measured the proportion (Prop.) of matching orders that pass the two pruning techniques and the wall-clock execution time of generating the best (incremental) execution plan(s).
Table~\ref{tab:exp_best_execution_plan_search} reports the average results for every $n$.
The pruning techniques were effective.
The time of generating the best execution plans was very short compared to enumeration.

\begin{table}[!t]
    \centering
    \caption{Efficiency of Best Execution Plan Generation\label{tab:exp_best_execution_plan_search}}
    \begin{tabular}{@{}ccccccccc@{}}
        \toprule
                                & \textbf{n} & \textbf{4} & \textbf{5} & \textbf{6} & \textbf{7} & \textbf{8} & \textbf{9} & \textbf{10} \\ \midrule
        \multirow{2}{*}{B-BENU} & Prop. (\%) & 8.3        & 5.0        & 2.4        & 0.7        & 0.2        & 0.1        & 0.01        \\
                                & Time (s)   & 0.2        & 0.2        & 0.3        & 0.4        & 0.7        & 2.7        & 27.1        \\ \midrule
        \multirow{2}{*}{S-BENU} & Prop. (\%) & 75.6       & 38.8       & 17.0       & 6.5        & 2.2        & 0.6        & 0.1         \\
                                & Time (s)   & 0.3        & 0.3        & 0.4        & 0.6        & 1.2        & 3.6        & 22.5        \\ \midrule
    \end{tabular}
\end{table}

\textbf{Exp-2: Execution Plan Optimizations}.
We evaluated the effectiveness of the execution plan optimizations proposed in Section~\ref{subsec:execution_plan_optimization} on B-BENU and S-BENU in Fig.~\ref{fig:effects_of_optimization_techniques}.
The X-axis represents execution plans optimized from the raw plan with more optimizations.
As the compression would negate some optimization techniques, we only used the compressed execution plan for Q5.
For the \emph{wk} dataset, we used 75\% of it as the initial graph and generated 5 time steps with 1M delta edges per time step.
Optimization 1 was effective for Q4 and Q5' where it eliminated common subexpressions.
Optimization 2 reduced the execution time in all cases by up to an order of magnitude.
It promoted INT instructions to outer loops in all of them.
Optimization 3 was effective for Q2 and Q5 where triangles were repeatedly enumerated by two INT instructions.

\begin{figure}[!t]
    \centering
    \subfloat[B-BENU\label{fig:effects_of_optimization_techniques_on_BENU}]{\includegraphics[height=5.3em]{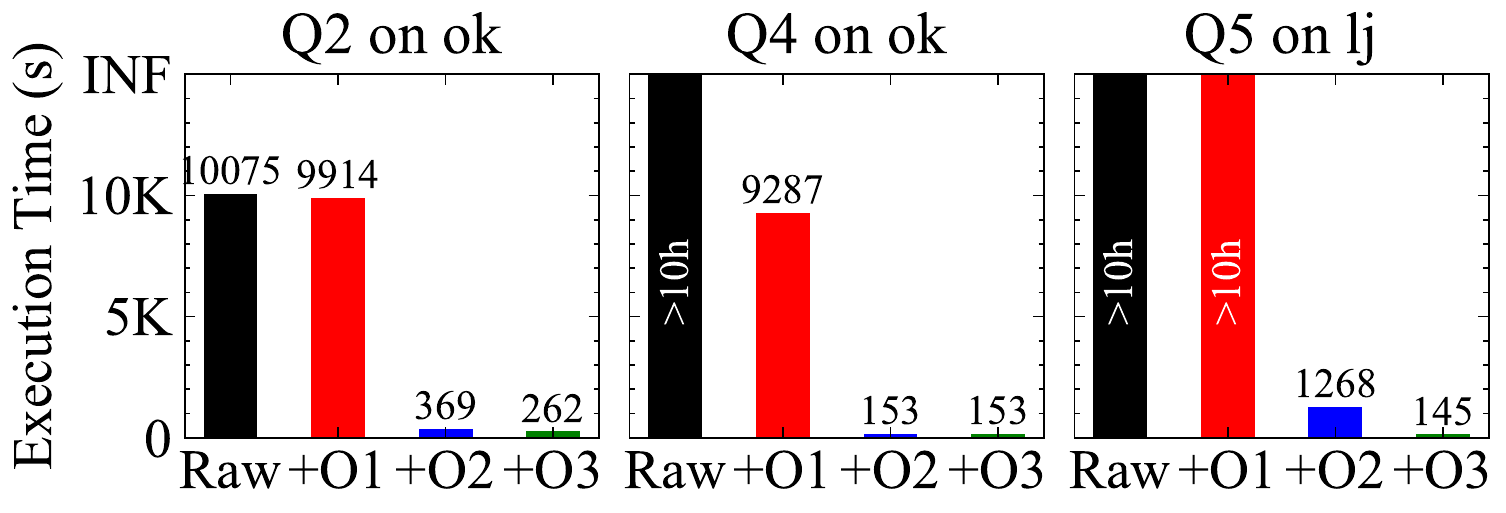}}
    \hfil
    \subfloat[S-BENU\label{fig:effects_of_optimization_techniques_on_SBENU}]{\includegraphics[height=5.3em]{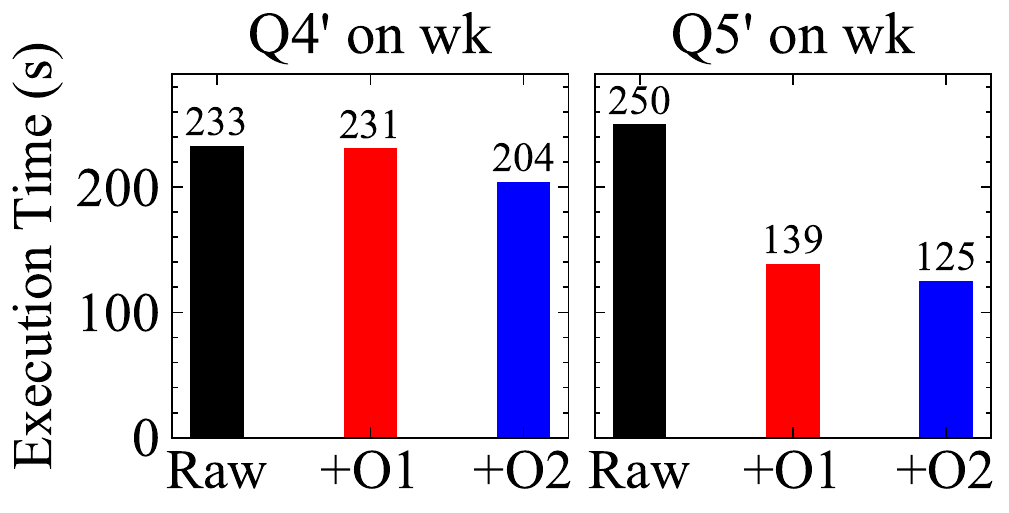}}
    \caption{Effects of optimization techniques.\label{fig:effects_of_optimization_techniques}}
\end{figure}

\textbf{Exp-3: Local Database Cache}.
We evaluated the effects of the capacity of the local database cache in Fig.~\ref{fig:effects_of_local_database_cache_capacity}.
The cache capacity is relative to the data graph (B-BENU) or the initial data graph (S-BENU).
The network communication cost and the execution time are relative to the corresponding cases with the 10\% relative cache capacity.
We evaluated B-BENU with Q4 and Q5 on \emph{ok}.
We evaluated S-BENU with Q3', and Q4' on \emph{ld} with the 80\% initial graph and 10 time steps (1M delta edges per time step).
Most pattern graphs (Q4, Q5, and Q3') were sensitive to the cache capacity.
The average cache hit rates increased quickly as the cache capacity grew.
Correspondingly, the communication cost and the execution time decreased quickly.
The DB cache was an effective technique to improve the efficiency in most cases.

\begin{figure}[!t]
    \centering
    \subfloat[Cache Hit Rate]{\includegraphics[width=0.32\columnwidth]{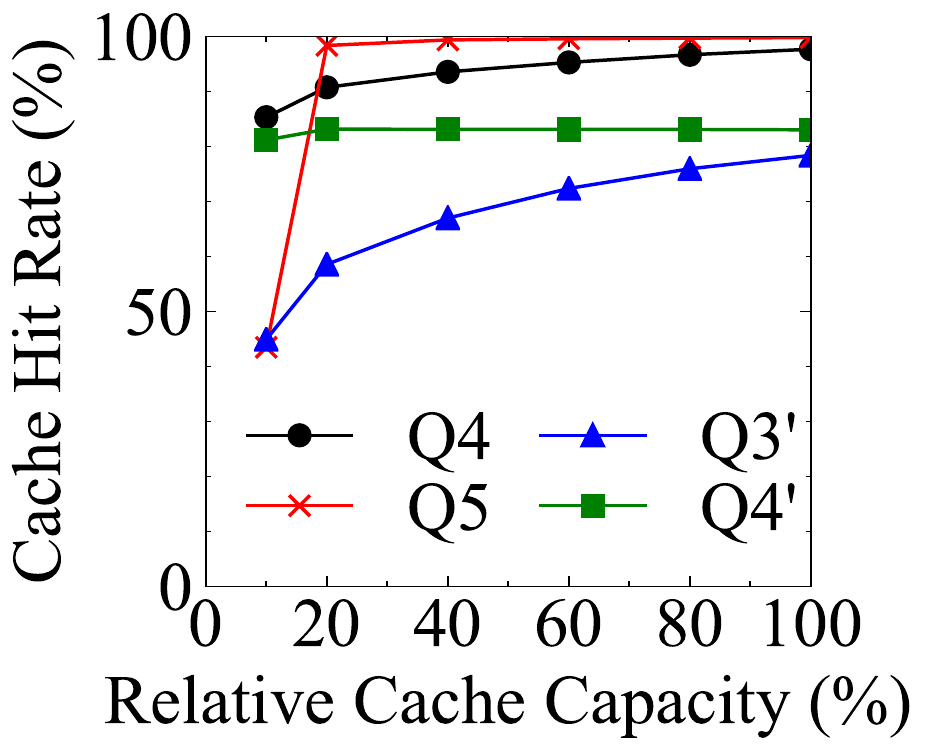}}
    \subfloat[Communication]{\includegraphics[width=0.32\columnwidth]{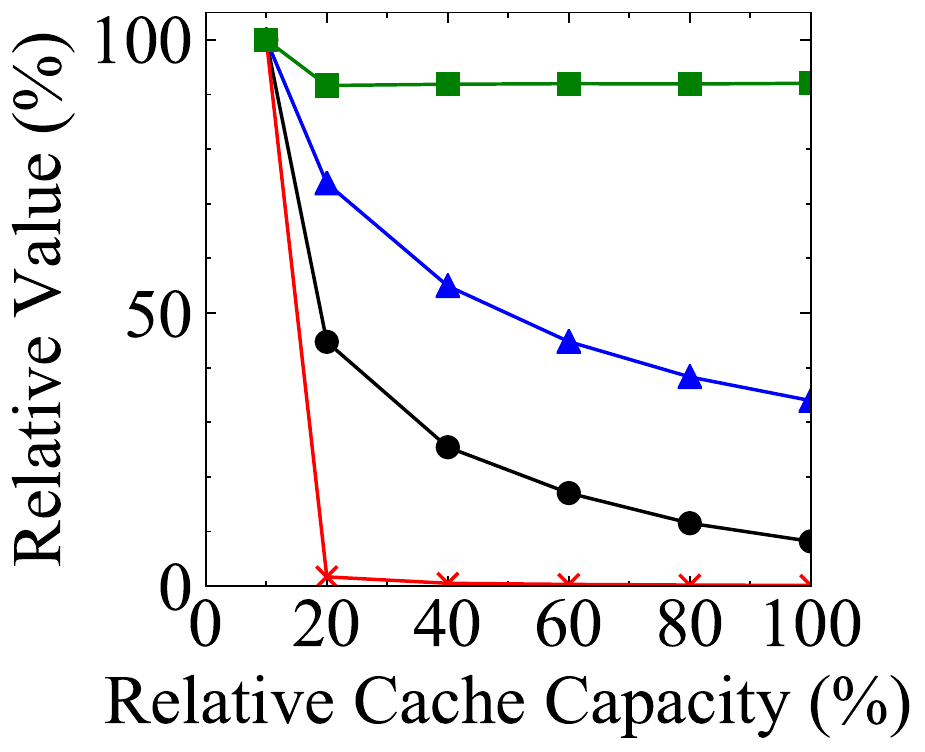}}
    \subfloat[Execution Time]{\includegraphics[width=0.32\columnwidth]{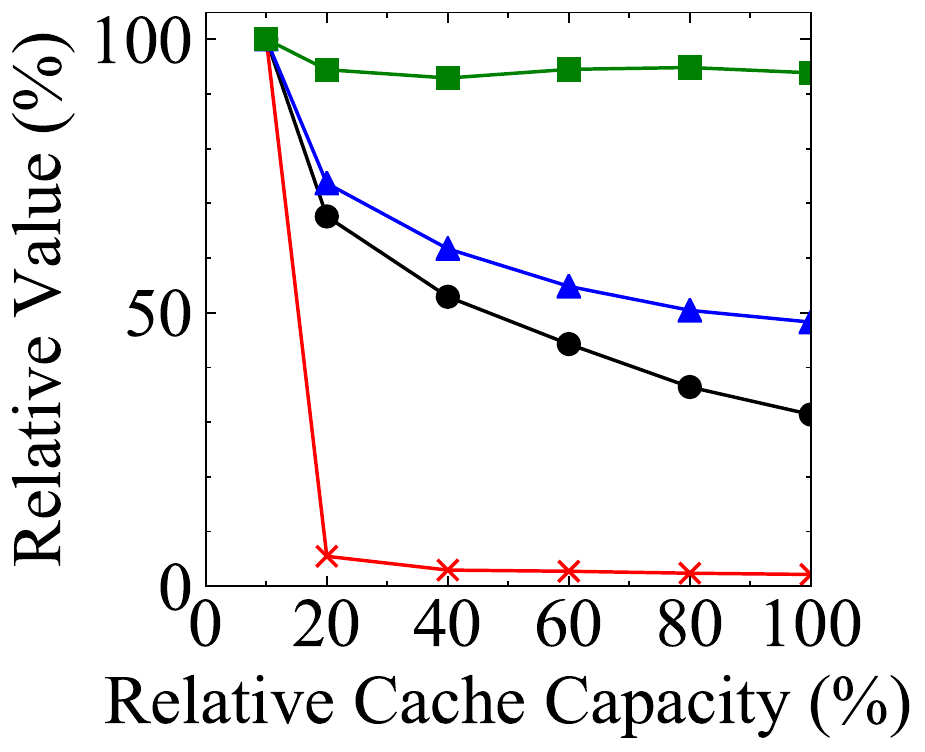}}
    \caption{Effects of the local database cache capacity.\label{fig:effects_of_local_database_cache_capacity}}
\end{figure}

\textbf{Exp-4: Task Splitting}.
To evaluate the effects of task splitting, we ran B-BENU and S-BENU with different split thresholds $\theta$.
For B-BENU, we measured the execution time of 16 reducers.
For S-BENU, we measured the execution time of 16 executors spent on executing local search tasks (Line 11 to 21 in Algorithm~\ref{alg:SBENU_framework}).
Fig.~\ref{fig:task_splitting_evaluation} shows the distribution of their execution time with different thresholds.
We ran the \emph{ld} dataset with 80\% of it as the initial graph and 1 time step of 10M delta edges.
The task splitting technique was much more effective on B-BENU than on S-BENU.
The technique made workloads more balanced among executors and working threads.
However, if the threshold was too low, the execution time increased due to more overheads.

\begin{figure}[!t]
    \centering
    \subfloat[B-BENU, Q5, \emph{ok}]{\includegraphics[width=0.4\columnwidth]{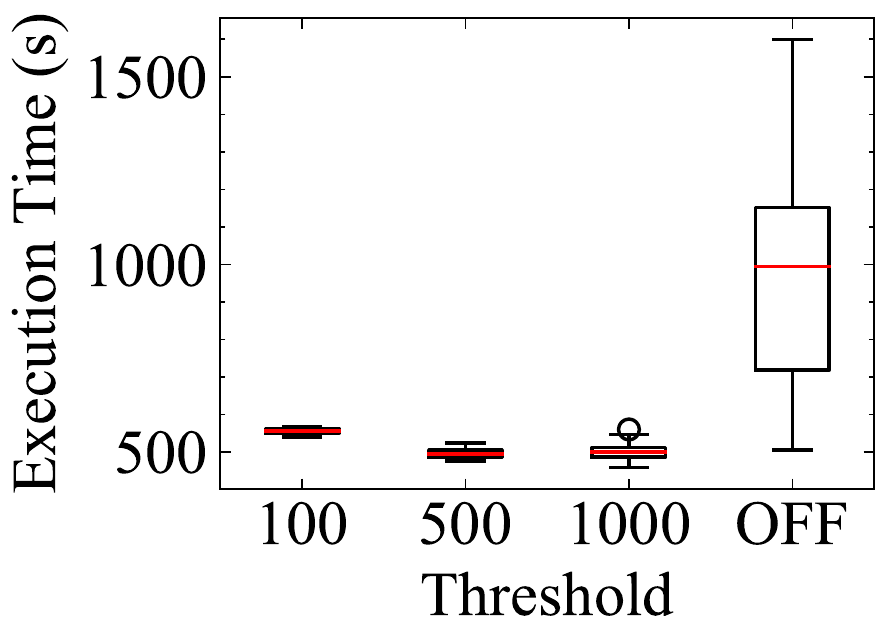}}
    \hfil
    \subfloat[S-BENU, Q3', \emph{ld}]{\includegraphics[width=0.4\columnwidth]{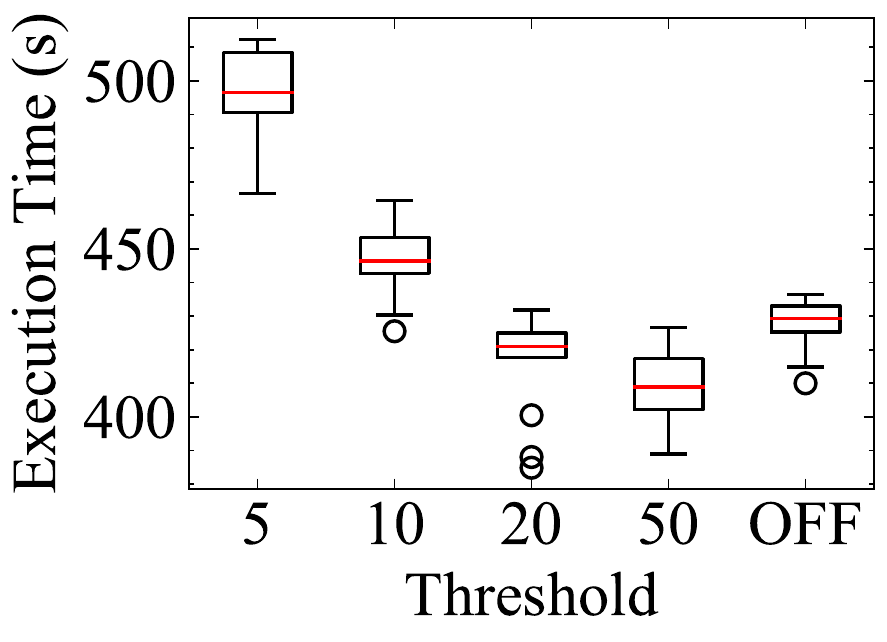}}
    \caption{Effects of the task splitting technique.\label{fig:task_splitting_evaluation}}
\end{figure}

\subsection{Comparing B-BENU with the State-of-the-art\label{sec:comparing_benu_with_state_of_the_art}}

We compared B-BENU with the state-of-the-art MapReduce-based method CBF \cite{CBF} and the worst-case-optimal-join-based BiGJoin \cite{BiGJoin}.
We turned on all the compression and optimization techniques provided with the algorithms.
We reported the wall-clock time spent on pure enumeration as the execution time, not including the time spent on join/execution plan generation and output.

\textbf{Exp-5: Comparison with CBF}.
CBF is the state-of-the-art algorithm in MapReduce.
We ran CBF with 12 mappers/reducers per worker machine and allocated 4 Gbytes memory to each mapper/reducer.
The results were reported in Table \ref{tab:perf-compare-cbf}.
Nearly in all the cases except Q5 on \emph{fs}, B-BENU ran quicker than CBF with acceptable communication costs.
In several cases like Q2 on \emph{ok}/\emph{uk}, Q4 on \emph{ok} and Q6 on \emph{lj}/\emph{ok}/\emph{uk}, B-BENU was up to 10$\times$ quicker than CBF.
The hard test cases Q7 to Q9 shared the same core structure, i.e. the chordal square (shown with bold edges in Fig.~\ref{fig:pattern_graphs}).
The core structure had more than 2 billion matches in all data graphs (Table~\ref{tab:statistics_of_typical_pattern_graphs}).
CBF had to shuffle the clique index and the matching results of the core structure during the preparation of partition files for hash-assembly.
Shuffling many key-value pairs was costly and made Hadoop throw the shuffle error exception in some cases.
B-BENU ran smoothly in those cases.
For the cases of Q7/Q8/Q9 on \emph{uk}, the core structure had 2.7 trillion matches.
Neither B-BENU nor CBF could finish in 10 hours.

\begin{table*}[!t]
    \centering
    \caption{Performance Comparison with CBF}
    \label{tab:perf-compare-cbf}
    \scriptsize
    \begin{tabular}{ccccccccccc}
        \toprule
        \multirow{2}{*}{\textbf{Dataset}} & \multicolumn{2}{c}{\textbf{Q1}} & \multicolumn{2}{c}{\textbf{Q2}} & \multicolumn{2}{c}{\textbf{Q3}} & \multicolumn{2}{c}{\textbf{Q4}} & \multicolumn{2}{c}{\textbf{Q5}}                                                                                                            \\ 
                                          & CBF                             & B-BENU                          & CBF                             & B-BENU                          & CBF                             & B-BENU             & CBF                 & B-BENU             & CBF                 & B-BENU             \\ \midrule
        as                                & 270/3G                          & \textbf{119/6G}                 & 167/26G                         & \textbf{69/6G}                  & 239/3G                          & \textbf{92/7G}     & 158/26G             & \textbf{68/5G}     & 356/1G              & \textbf{131/6G}    \\ 
        lj                                & 396/11G                         & \textbf{183/16G}                & 662/210G                        & \textbf{102/16G}                & 348/11G                         & \textbf{138/17G}   & 656/207G            & \textbf{117/15G}   & 190/5G              & \textbf{128/14G}   \\ 
        ok                                & 2942/29G                        & \textbf{859/30G}                & 1465/512G                       & \textbf{139/28G}                & 1446/28G                        & \textbf{425/29G}   & 1507/508G           & \textbf{139/26G}   & 1024/14G            & \textbf{595/29G}   \\ 
        uk                                & \textgreater{}7200s             & \textbf{2131/90G}               & \textgreater{}7200s             & \textbf{412/81G}                & \textgreater{}7200s             & \textbf{1221/93G}  & \textgreater{}7200s & \textbf{930/85G}   & \textgreater{}7200s & \textbf{3549/103G} \\ 
        fs                                & \textgreater{}41555s            & \textbf{16622/416G}             & CRASH                           & \textbf{5008/472G}              & \textgreater{}10547s            & \textbf{4219/391G} & \textgreater{}7200s & \textbf{1543/371G} & \textbf{2088/137G}  & 4484/392G          \\ \bottomrule
    \end{tabular}

    \vspace{0.5em}

    \begin{tabular}{ccccccccc}
        \toprule
        \multirow{2}{*}{\textbf{Dataset}} & \multicolumn{2}{c}{\textbf{Q6}} & \multicolumn{2}{c}{\textbf{Q7}} & \multicolumn{2}{c}{\textbf{Q8}} & \multicolumn{2}{c}{\textbf{Q9}}                                                                                        \\ 
                                          & CBF                             & B-BENU                          & CBF                             & B-BENU                          & CBF                  & B-BENU             & CBF                 & B-BENU             \\ \midrule
        as                                & 288/4G                          & \textbf{68/4G}                  & CRASH                           & \textbf{1188/6G}                & CRASH                & \textbf{7632/7G}   & CRASH               & \textbf{315/5G}    \\ 
        lj                                & 1000/20G                        & \textbf{108/12G}                & CRASH                           & \textbf{9318/20G}               & \textgreater{}16710s & \textbf{6684/17G}  & \textgreater{}7200s & \textbf{2111/16G}  \\ 
        ok                                & 2556/48G                        & \textbf{143/22G}                & CRASH                           & \textbf{2327/31G}               & \textgreater{}7200s  & \textbf{2974/29G}  & \textgreater{}7200s & \textbf{712/28G}   \\ 
        uk                                & 16488/131G                      & \textbf{1090/39G}               & \textgreater{}10h               & \textgreater{}10h               & FAIL$^*$             & \textgreater{}10h  & FAIL$^*$            & \textgreater{}10h  \\ 
        fs                                & 18472/691G                      & \textbf{1349/314G}              & \textgreater{}11272s            & \textbf{4509/424G}              & \textgreater{}10282s & \textbf{4113/362G} & \textgreater{}7200s & \textbf{2464/350G} \\ \bottomrule
    \end{tabular}

    \begin{flushleft}
        \scriptsize
        $^+$ In each cell, the first number is the wall-clock execution time (unit: second), and the second number is the cumulative communication cost (unit: byte). \\ $^\times$ The quickest algorithm in each case is marked with bold font. $^*$ CBF failed when building the clique index of four vertices.
    \end{flushleft}

\end{table*}

\textbf{Exp-6: Comparison with BiGJoin}
BiGJoin \cite{BiGJoin} is the state-of-the-art worst-case-optimal algorithm.
We compared B-BENU with it on the pattern graphs that BiGJoin had specially optimized.
BiGJoin (\url{https://github.com/frankmcsherry/dataflow-join/}) was written with the Timely dataflow system in Rust.
In BiGJoin, the batch size was 100000, and each worker machine was deployed with 12 working processes (one process per core).
We compared B-BENU with both the shared-memory version (BiGJoin(S)) and the distributed version (BiGJoin(D)) of BiGJoin in Table \ref{tab:comparison-with-bigjoin}.
Since BiGJoin used a different communication mechanism from MapReduce, we did not report the communication costs.
On \emph{ok}, B-BENU ran quicker than both of BiGJoin(D) and BiGJoin(S) with complex pattern graphs.
On \emph{fs}, BiGJoin(S) failed due to out of memory exception, while B-BENU ran quicker than BiGJoin(D) in all cases.

\begin{table}[!t]
    \centering
    \caption{Execution Time Comparison with BiGJoin}
    \label{tab:comparison-with-bigjoin}
    \scriptsize
    \begin{tabular}{p{1em}cccccc}
        \toprule
        \textbf{\textit{G}} & \textbf{Algorithm} & \textbf{Triangle} & \textbf{Clique4}   & \textbf{Clique5}   & \textbf{Q4}        & \textbf{Q5}        \\ \midrule
        \multirow{3}{*}{ok} & BiGJoin(S)         & \textbf{53}       & 111                & 651                & 608                & OOM                \\ 
                            & BiGJoin(D)         & 130               & OOM                & OOM                & \textgreater{}7200 & OOM                \\ 
                            & B-BENU             & 93                & \textbf{99}        & \textbf{129}       & \textbf{139}       & \textbf{595}       \\ \midrule
        \multirow{3}{*}{fs} & BiGJoin(S)         & OOM               & OOM                & OOM                & OOM                & OOM                \\ 
                            & BiGJoin(D)         & 1749              & \textgreater{}7200 & \textgreater{}7200 & \textgreater{}7200 & \textgreater{}7200 \\ 
                            & B-BENU             & \textbf{1229}     & \textbf{1239}      & \textbf{1251}      & \textbf{1543}      & \textbf{4484}      \\ \bottomrule
    \end{tabular}
    \begin{flushleft}
        \scriptsize
        Unit: second.
    \end{flushleft}
\end{table}

\subsection{Comparing S-BENU with the State-of-the-art\label{sec:comparing_sbenu_with_state_of_the_art}}

\begin{figure}[!t]
    \centering
    \includegraphics[width=\columnwidth]{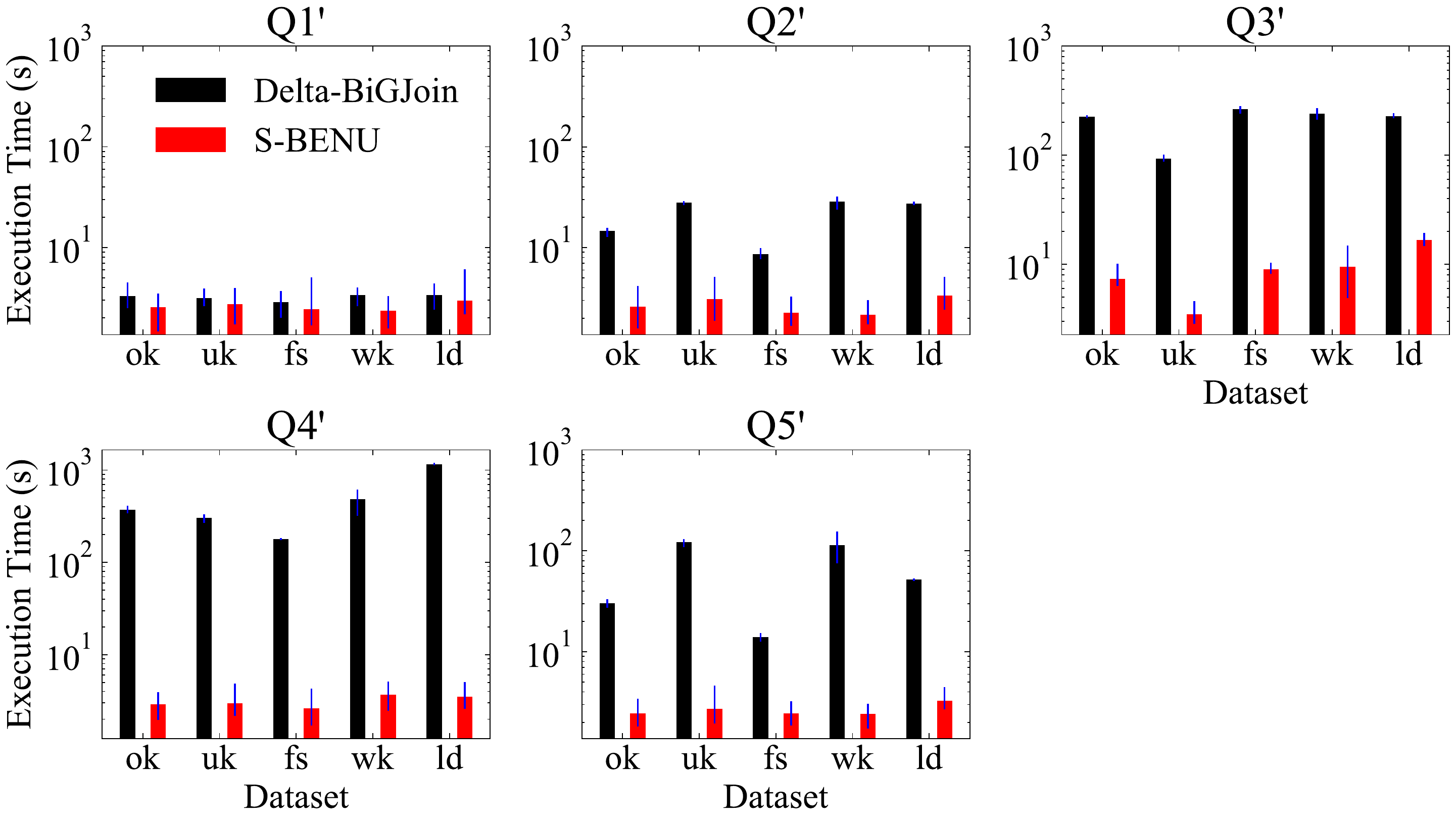}
    \caption{Performance comparison of Delta-BiGJoin and S-BENU.}
    \label{fig:performane_comparison_of_Delta_BiGJoin_and_SBENU}
\end{figure}

We compared the performance of S-BENU with the state-of-the-art distributed continuous subgraph enumeration algorithm Delta-BiGJoin \cite{BiGJoin}.
For Delta-BiGJoin, we deployed 12 worker processes (one process per core) on each worker machine.
For all datasets, we used 20\% initial graph and generated 10 time steps of 20K delta edges per time step.
For the \emph{ok}, \emph{uk}, and \emph{fs} datasets that do not have timestamps attached on edges, we picked edges randomly from the remaining 80\% graph to generate update operations.
We measured the wall-clock execution time of S-BENU and Delta-BiGJoin spent on each time step.
The execution time did not include outputing because it was independent of subgraph enumeration.
Fig.~\ref{fig:performane_comparison_of_Delta_BiGJoin_and_SBENU} shows the average time of 10 time steps with error bars indicating the maximal and the minimal.
S-BENU outperformed Delta-BiGJoin in all cases, by up to two orders of magnitude.

Compared to S-BENU, Delta-BiGJoin suffered from high communication costs, which is caused by shuffling intermediate matching results.
Taking enumerating Q4' on \emph{ld} as the example, Delta-BiGJoin shuffled the matching results of the partial pattern graphs $u_2$-$u_5$-$u_1$ and $u_3$-$u_5$-$u_1$-$u_4$.
They had 35M and 197M matches respectively, causing high communication costs.

\subsection{Machine Scalability\label{sec:scalability}}

We tested the machine scalability of B-BENU and S-BENU by varying numbers of worker nodes in the cluster.
For S-BENU, we used \emph{ld} with 80\% initial graph and 1 time step.

\begin{figure}[!t]
    \centering
    \subfloat[B-BENU on \emph{ok}]{\includegraphics[height=8em]{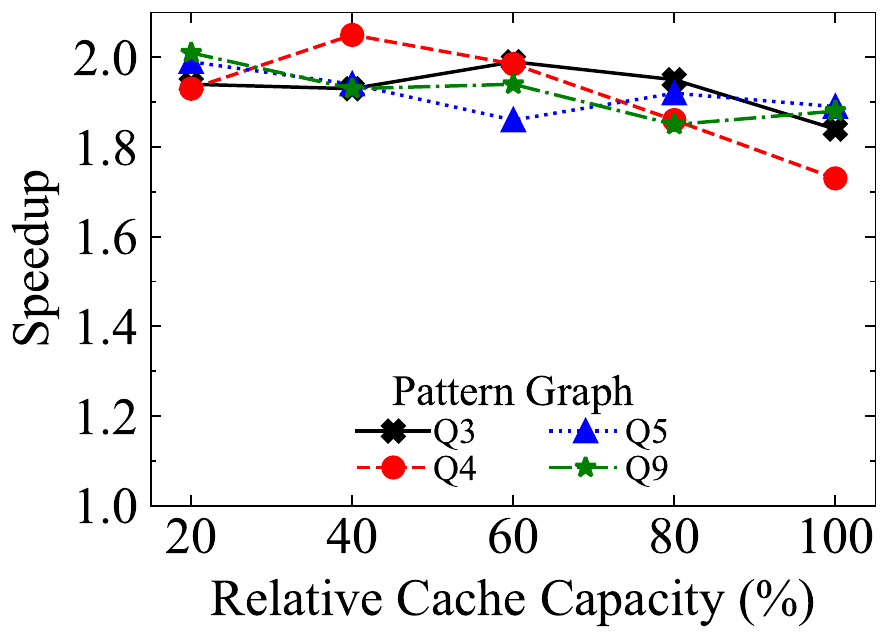}}
    \hfil
    \subfloat[S-BENU on \emph{ld}\label{fig:effects_of_hyper_parameters_on_scalability_sbenu}]{\includegraphics[height=8em]{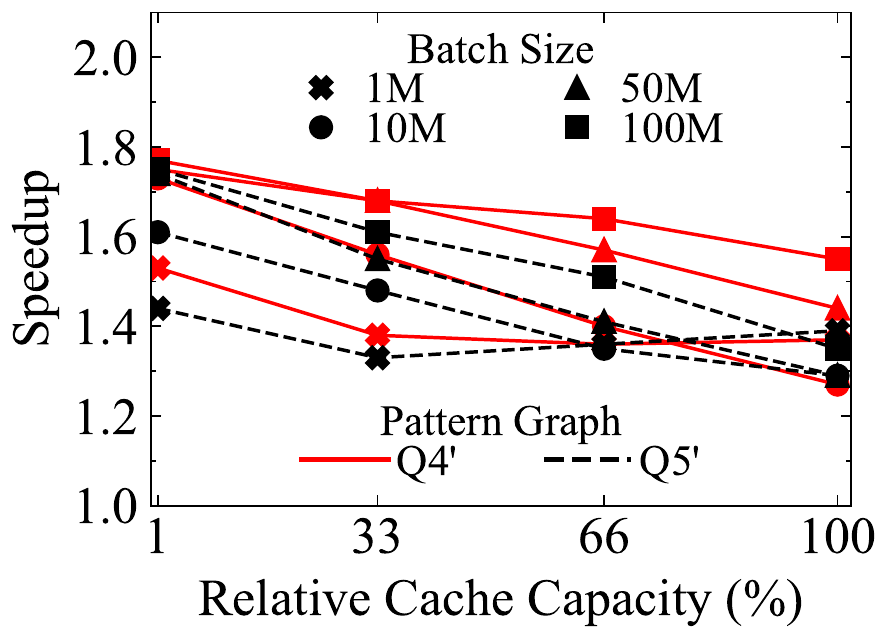}}
    \caption{Effects of hyper parameters on scalability.}
    \label{fig:effects_of_hyper_paramters_on_scalability}
\end{figure}

The cache capacity is a common hyperparameter of B-BENU and S-BENU.
We measured the execution time under different capacities with 8 and 16 workers, and we reported the speedups in Fig.~\ref{fig:effects_of_hyper_paramters_on_scalability}.
The cache capacities were relative to the (initial) data graph.
B-BENU and S-BENU achieved better scalability with a smaller cache.
The phenomenon is caused by the warm-up phase of the cache.
The cache in each machine had to go through a warm-up phase before it achieved a stable hit rate.
In the warm-up phase, there were many cache misses.
When B-BENU/S-BENU processed the same input with more workers, more database queries were conducted during the warm-up phase.
It caused more cache misses and increased the total execution cost.
Taking S-BENU on Q4' with 1M batch size as an example, the total serial execution time of all tasks with 16 workers increased by 16\% (1\% capacity) and 33\% (100\% capacity) compared to 8 workers.
The increased execution cost harmed the scalability.

The batch size is a hyperparameter specific to S-BENU.
Fig.~\ref{fig:effects_of_hyper_parameters_on_scalability_sbenu} shows that S-BENU achieved better scalability with larger batch sizes.
With a larger batch, the number of conducted database queries was larger and the proportion of queries that were conducted during the warm-up phase of the cache became smaller.
Taking Q5$'$ with 33\% cache capacity as an example, the total serial execution time of all tasks with 16 workers increased by 35\% (1M batch) and 18\% (100M batch) compared to 8 workers.

\begin{figure}[!t]
    \centering
    \subfloat[B-BENU on \emph{ok}\label{fig:machine_scalability_benu}]{\includegraphics[height=8em]{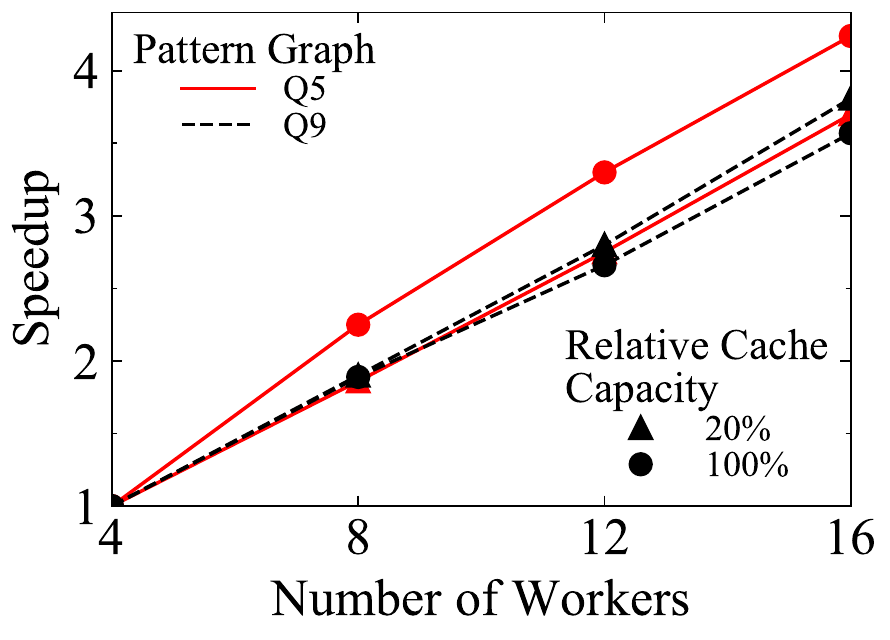}}
    \hfil
    \subfloat[S-BENU on \emph{ld}\label{fig:machine_scalability_sbenu}]{\includegraphics[height=8em]{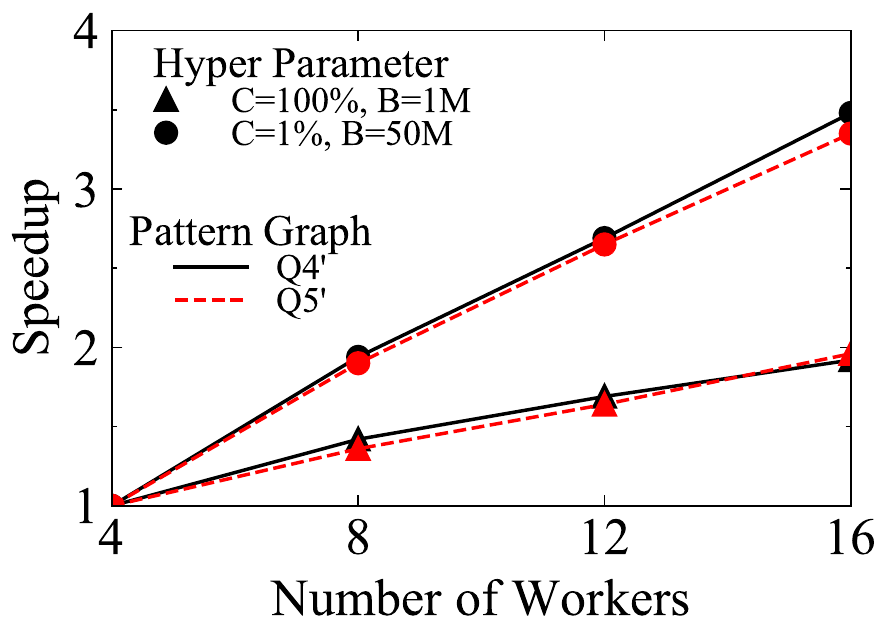}}
    \caption{Machine scalability.\label{fig:machine_scalability}}
\end{figure}

Fig.~\ref{fig:machine_scalability} reports the machine scalability of B-BENU and S-BENU by varying the number of worker machines from 4 to 16.
B-BENU shew the near-linear scalability with both big caches and small caches.
In the case of Q5 with 20\% relative cache capacity in Fig.~\ref{fig:machine_scalability_benu}, B-BENU shew a bit superlinear speedup because there was a straggler in the case of 4 workers, making the execution time longer.
In the legend of Fig.~\ref{fig:machine_scalability_sbenu}, \emph{C} represents the relative cache capacity and \emph{B} represents the batch size.
S-BENU shew the near-linear scalability with small caches and big batches.
It shew the sub-linear scalability with big caches and small batches, but the speedup curve was still linear.

\section{Related Work}


\subsection{Serial Subgraph Enumeration}

Most of the serial subgraph matching methods work with labeled graphs and follow the backtracking-based framework \cite{BacktrackingFramework}.
They differ in how to determine the matching order and the candidate sets of pattern vertices.
GraphQL\cite{GraphQL} and SPath \cite{SPath} match pattern vertices with infrequent labels and paths first.
TurboISO \cite{TurboISO} uses the candidate region to dynamically determine the matching order and candidate sets.
CFL-Match \cite{CFLMatch} proposes the core-forest-leaf decomposition to get matching orders with Cartesian operations postponed.
DAF \cite{DAF} uses a DAG-graph-based dynamic programming to get a candidate space structure and follows an adaptive matching order.
CECI \cite{CECI} divides the data graph into embedding clusters and builds a compact index for each cluster to prune the search space.
Those in-memory algorithms cannot work with data graphs larger than the memory.
DUALSIM \cite{DUALSIM} uses the disk to handle large data graphs and proposes the dual approach to reduce the number of disk reads.
However, the computing power of a single machine limits its performance.

\subsection{Distributed Subgraph Enumeration}
Based whether a method shuffles intermediate results, the existing distributed subgraph enumeration methods can be divided into DFS-style and BFS-style.

The DFS-style methods do not shuffle intermediate results.
Instead, they shuffle the data graphs.
Afrati et al. \cite{UllmanOneRound} replicate different parts of data graphs to every reducer in MapReduce and use one-round multiway join on the reducer side to enumerate subgraphs.
The amount of replicated edges grows quickly as the pattern graph becomes complex \cite{TwinTwigJoinJournalVersion}.
QFrag \cite{QFrag} replicates the whole data graph in the memory of every worker node and enumerates subgraphs with serial subgraph matching methods in parallel.
The memory capacity limits the scale of the data graph that it can handle.
RADS \cite{RADS} partitions the data graph among worker nodes.
It first enumerates subgraphs residing in the local partition of each node.
It then exchanges and verifies undetermined edges among partitions in a region-grouped manner via MPI to enumerate subgraphs cross partitions.
However, RADS relies on MPI to implement the node-to-node communication.
It does not consider the fault tolerance, which is essential in distributed computing.
CECI \cite{CECI} constructs embedding clusters from the data graph, distributes the clusters among worker nodes, and enumerates subgraphs from the clusters in parallel.
Work stealing via MPI is used to balance workloads.
The memory of each node limits the size of embedding clusters that it can handle.

The BFS-style methods follow a join-based framework.
They decompose the pattern graph into join units, enumerate partial matches of join units, and join partial matches together to get matches for the whole pattern graph.
They implement the join framework with a distributed dataflow engine like MapReduce and Timely that transparently support the fault tolerance.
The join-based methods have to shuffle intermediate results during the join.
To limit their sizes, varieties of join units (Edge \cite{BiGJoin}, Star \cite{STWig} \cite{SEED} \cite{PSgL}, TwinTwig \cite{TwinTwigJoin}, Clique \cite{SEED} and Crystal \cite{CBF}) and join frameworks (left-deep join \cite{TwinTwigJoin,TwinTwigJoinJournalVersion}, two-way bushy join \cite{SEED}, hash-assembly \cite{CBF}, multiway \cite{STWig} and worst-case optimal join \cite{BiGJoin}) are proposed.
Lai et al. \cite{SurveyOnTimely} experimentally survey the existing join-based methods with a unified implementation on the Timely dataflow engine.
They find that no method could win all scenarios.
They propose a practical guide to select a suitable method for specific inputs.
Qiao et al. \cite{CBF} propose the VCBC compression to compress the (partial) matching results.

\subsection{Serial Continuous Subgraph Enumeration}

Given a batch of edge updates, IncIsoMat \cite{IncIsoMat,IncIsoMatJournal} first finds out the local neighborhood affected by
the batch and then conducts subgraph isomorphism in the neighborhood.
IncIsoMat compares the matching results before and after the update to discover incremental matches.
However, IncIsoMat has to conduct subgraph isomorphism from scratch for every update.
It will do redundant computation if the neighborhoods affected by two updates are overlapped.

One way to avoid such redundancy is to maintain matching results of the pattern graph in memory as the data graph evolves.
SJ-Tree \cite{SJ-Tree} stores partial matches of the pattern graph in a binary join tree with single edges and 2-edge paths as join units.
SJ-Tree stores partial matches of a tree node in a hash table.
When new edges arrive, SJ-Tree joins new edges with the stored partial matches, avoiding enumerating matches from scratch.
However, the hash table is a memory-consuming data structure.
If there are many partial matches, the memory usage of SJ-Tree will be high.
To store matching results compactly, TurboFlux \cite{TurboFlux} proposes the data-centric graph (DCG) structure.
Given an edge update, TurboFlux transmits the states of the edges in DCG and detects incremental matches during the transition.
However, the edge transition model is serial.
It can only process update edges one by one, limiting the throughput of TurboFlux.

The other way to avoid redundancy is to compute incremental matches from update edges directly.
GraphFlow \cite{GraphFlow} models the continuous subgraph enumeration as the incremental multiway-join view maintainance problem.
It computes the incremental matches by joining update edges with the existing data graph together, guided by the delta rule technique \cite{DeltaRules}.
It adopts a variant of GenericJoin \cite{GenericJoin} as the underlying multi-way join framework.

Though the serial methods have little execution overheads, the computing power and memory capacity of a single node prevent the serial methods from handling big graphs and large update batches.

\subsection{Distributed Continuous Subgraph Enumeration}

D-IDS \cite{D-IDS} prunes the data graph with the maximal dual simulation technique and enumerates matches in subgraphs affected by the update update in parallel.
It maintains the matching results in memory in a distributed way to detect disappearing matches.
When the data graph is big and the pattern graph is complex, the matching results may exceed the memory capacity of a single node.
Delta-BiGJoin \cite{BiGJoin} is the variant of BiGJoin for dynamic graphs.
It partitions and stores the data graph in memory among all worker nodes.
Delta-BiGJoin does not maintain any matching result, making the memory usage controllable.
Instead, it treats the continuous subgraph enumeration as the incremental view maintenance problem in the relational database.
Given a batch of edge updates, Delta-BiGJoin computes the incremental matches via a group of multi-way join queries, guided by the delta rule technique \cite{DeltaRules}.
Delta-BiGJoin uses BiGJoin as its underlying multiway join implementation.
Under the insertion-only workloads, Delta-BiGJoin is worst-case optimal in both computation and communication.
D-IDS and Delta-BiGJoin are general-purpose methods.
Some methods optimize for specific pattern graphs, like vertex-level matching \cite{SSDJournal} and cycles \cite{GraphS}.

\section{Conclusion}
In this paper, we studied the distributed (continuous) subgraph enumeration problem.
The state-of-the-art distributed methods are based on distributed join that has to shuffle intermediate results.
When the data graph is big and the pattern graph is complex, the scale of intermediate results can be huge.
To overcome the drawback, we proposed a backtracking-based framework \emph{Batch-BENU} with two features:
(1) shuffling data graph instead of intermediate results, (2) on-demand shuffle.
Batch-BENU stored the data graph in a distributed key-value database and queried its adjacency sets on demand driven by backtracking-based execution plans.
Given a pattern graph, we proposed a search-based method to generate the best execution plan for it.
We also proposed three optimization techniques (common subexpression elimination, instruction reordering, and triangle cache) to reduce the execution costs of execution plans.
To support dynamic data graphs, we proposed the Streaming-BENU framework.
Streaming-BENU solved the continuous subgraph enumeration problem by enumerating incremental pattern graphs at each time step.
We proposed efficient implementations for Batch-BENU and Streaming-BENU on Hadoop and Spark, respectively.
We developed the local database cache technique and the task splitting technique to improve performance.
We also discussed the data structure to store the dynamic data graph in the database and in the cache.
Extensive experiments verified the efficiency of Batch-BENU and Streaming-BENU.
Batch-BENU and Streaming-BENU outperformed the state-of-the-art distributed methods by up to one and two orders of magnitude, respectively.

In the future, we shall explore 1) extending Batch-BENU and Streaming-BENU to property graphs, 2) finding a more accurate model to estimate the scale of matching results, and 3) generalizing the triangle cache technique to cliques.

The source code of Batch-BENU and Streaming-BENU is available at \url{https://github.com/PasaLab/BENU}.

\ifCLASSOPTIONcompsoc
  \section*{Acknowledgments}
  
  This work is funded in part by National Key R\&D Program of China [grant number 2019YFC1711000]; 
  China NSF Grants [grant number U1811461];
  Jiangsu Province Industry Support Program [grant number BE2017155, BK20170651];
  Collaborative Innovation Center of Novel Software Technology and Industrialization;
  and the program B for Outstanding PhD candidate of Nanjing University.
  
\else
  \section*{Acknowledgment}
\fi

\ifCLASSOPTIONcaptionsoff
  \newpage
\fi



\bibliographystyle{IEEEtran}
\bibliography{benu-ref}

\begin{thebibliography}{10}
\providecommand{\url}[1]{#1}
\csname url@samestyle\endcsname
\providecommand{\newblock}{\relax}
\providecommand{\bibinfo}[2]{#2}
\providecommand{\BIBentrySTDinterwordspacing}{\spaceskip=0pt\relax}
\providecommand{\BIBentryALTinterwordstretchfactor}{4}
\providecommand{\BIBentryALTinterwordspacing}{\spaceskip=\fontdimen2\font plus
\BIBentryALTinterwordstretchfactor\fontdimen3\font minus
  \fontdimen4\font\relax}
\providecommand{\BIBforeignlanguage}[2]{{%
\expandafter\ifx\csname l@#1\endcsname\relax
\typeout{** WARNING: IEEEtran.bst: No hyphenation pattern has been}%
\typeout{** loaded for the language `#1'. Using the pattern for}%
\typeout{** the default language instead.}%
\else
\language=\csname l@#1\endcsname
\fi
#2}}
\providecommand{\BIBdecl}{\relax}
\BIBdecl

\bibitem{Milo2002NetworkMS}
R.~Milo, S.~Shen-Orr, S.~Itzkovitz, N.~Kashtan, D.~Chklovskii, and U.~Alon,
  ``Network motifs: Simple building blocks of complex networks,''
  \emph{Science}, vol. 298, no. 5594, pp. 824--827, 2002.

\bibitem{DBLP:journals/bioinformatics/Przulj07}
N.~Przulj, ``Biological network comparison using graphlet degree
  distribution,'' \emph{Bioinform.}, vol.~23, no.~2, pp. 177--183, 2007.

\bibitem{DBLP:conf/wsdm/KairamWL12}
S.~R. Kairam, D.~J. Wang, and J.~Leskovec, ``The life and death of online
  groups: predicting group growth and longevity,'' in \emph{Proceedings of the
  Fifth International Conference on Web Search and Web Data Mining, {WSDM}
  2012, Seattle, WA, USA, February 8-12, 2012}, 2012, pp. 673--682.

\bibitem{DBLP:journals/pvldb/FanWWX15}
W.~Fan, X.~Wang, Y.~Wu, and J.~Xu, ``Association rules with graph patterns,''
  \emph{Proc. {VLDB} Endow.}, vol.~8, no.~12, pp. 1502--1513, 2015.

\bibitem{neo4j-fraud-detection}
G.~Sadowksi and P.~Rathle, ``Fraud detection: Discovering connections with
  graph databases,''
  \url{https://neo4j.com/whitepapers/fraud-detection-graph-databases/}, Neo4j,
  Tech. Rep., 2017.

\bibitem{GraphS}
X.~Qiu, W.~Cen, Z.~Qian, Y.~Peng, Y.~Zhang, X.~Lin, and J.~Zhou, ``Real-time
  constrained cycle detection in large dynamic graphs,'' \emph{Proc. {VLDB}
  Endow.}, vol.~11, no.~12, pp. 1876--1888, 2018.

\bibitem{cyber-security}
S.~Choudhury, L.~B. Holder, G.~C. Jr., A.~Ray, S.~Beus, and J.~Feo,
  ``Streamworks: a system for dynamic graph search,'' in \emph{Proceedings of
  the {ACM} {SIGMOD} International Conference on Management of Data, {SIGMOD}
  2013, New York, NY, USA, June 22-27, 2013}, 2013, pp. 1101--1104.

\bibitem{SEED}
L.~Lai, L.~Qin, X.~Lin, Y.~Zhang, and L.~Chang, ``Scalable distributed subgraph
  enumeration,'' \emph{Proc. {VLDB} Endow.}, vol.~10, no.~3, pp. 217--228,
  2016.

\bibitem{CBF}
M.~Qiao, H.~Zhang, and H.~Cheng, ``Subgraph matching: on compression and
  computation,'' \emph{Proc. {VLDB} Endow.}, vol.~11, no.~2, pp. 176--188,
  2017.

\bibitem{TurboISO}
W.~Han, J.~Lee, and J.~Lee, ``Turbo\({}_{\mbox{iso}}\): towards ultrafast and
  robust subgraph isomorphism search in large graph databases,'' in
  \emph{Proceedings of the {ACM} {SIGMOD} International Conference on
  Management of Data, {SIGMOD} 2013, New York, NY, USA, June 22-27, 2013},
  2013, pp. 337--348.

\bibitem{CFLMatch}
F.~Bi, L.~Chang, X.~Lin, L.~Qin, and W.~Zhang, ``Efficient subgraph matching by
  postponing cartesian products,'' in \emph{Proceedings of the 2016
  International Conference on Management of Data, {SIGMOD} Conference 2016, San
  Francisco, CA, USA, June 26 - July 01, 2016}, 2016, pp. 1199--1214.

\bibitem{DUALSIM}
H.~Kim, J.~Lee, S.~S. Bhowmick, W.~Han, J.~Lee, S.~Ko, and M.~H.~A. Jarrah,
  ``{DUALSIM:} parallel subgraph enumeration in a massive graph on a single
  machine,'' in \emph{Proceedings of the 2016 International Conference on
  Management of Data, {SIGMOD} Conference 2016, San Francisco, CA, USA, June 26
  - July 01, 2016}, 2016, pp. 1231--1245.

\bibitem{QFrag}
M.~Serafini, G.~D.~F. Morales, and G.~Siganos, ``Qfrag: distributed graph
  search via subgraph isomorphism,'' in \emph{Proceedings of the 2017 Symposium
  on Cloud Computing, SoCC 2017, Santa Clara, CA, USA, September 24-27, 2017},
  2017, pp. 214--228.

\bibitem{UllmanOneRound}
F.~N. Afrati, D.~Fotakis, and J.~D. Ullman, ``Enumerating subgraph instances
  using map-reduce,'' in \emph{29th {IEEE} International Conference on Data
  Engineering, {ICDE} 2013, Brisbane, Australia, April 8-12, 2013}, 2013, pp.
  62--73.

\bibitem{TwinTwigJoin}
L.~Lai, L.~Qin, X.~Lin, and L.~Chang, ``Scalable subgraph enumeration in
  mapreduce,'' \emph{Proc. {VLDB} Endow.}, vol.~8, no.~10, pp. 974--985, 2015.

\bibitem{TwinTwigJoinJournalVersion}
------, ``Scalable subgraph enumeration in mapreduce: a cost-oriented
  approach,'' \emph{{VLDB} J.}, vol.~26, no.~3, pp. 421--446, 2017.

\bibitem{snapnets}
J.~Leskovec and A.~Krevl, ``{SNAP Datasets}: {Stanford} large network dataset
  collection,'' \url{http://snap.stanford.edu/data}, Jun. 2014.

\bibitem{LAW-Datasets}
{Laboratory Web for Algorithmics}, ``Datasets,''
  \url{http://law.di.unimi.it/datasets.php}.

\bibitem{BiGJoin}
K.~Ammar, F.~McSherry, S.~Salihoglu, and M.~Joglekar, ``Distributed evaluation
  of subgraph queries using worst-case optimal and low-memory dataflows,''
  \emph{Proc. {VLDB} Endow.}, vol.~11, no.~6, pp. 691--704, 2018.

\bibitem{PSgL}
Y.~Shao, B.~Cui, L.~Chen, L.~Ma, J.~Yao, and N.~Xu, ``Parallel subgraph listing
  in a large-scale graph,'' in \emph{International Conference on Management of
  Data, {SIGMOD} 2014, Snowbird, UT, USA, June 22-27, 2014}, 2014, pp.
  625--636.

\bibitem{BENU-ICDE}
Z.~Wang, R.~Gu, W.~Hu, C.~Yuan, and Y.~Huang, ``{BENU:} distributed subgraph
  enumeration with backtracking-based framework,'' in \emph{35th {IEEE}
  International Conference on Data Engineering, {ICDE} 2019, Macao, China,
  April 8-11, 2019}, 2019, pp. 136--147.

\bibitem{SJ-Tree}
S.~Choudhury, L.~B. Holder, G.~C. Jr., K.~Agarwal, and J.~Feo, ``A selectivity
  based approach to continuous pattern detection in streaming graphs,'' in
  \emph{Proceedings of the 18th International Conference on Extending Database
  Technology, {EDBT} 2015, Brussels, Belgium, March 23-27, 2015}, 2015, pp.
  157--168.

\bibitem{TurboFlux}
K.~Kim, I.~Seo, W.~Han, J.~Lee, S.~Hong, H.~Chafi, H.~Shin, and G.~Jeong,
  ``Turboflux: {A} fast continuous subgraph matching system for streaming graph
  data,'' in \emph{Proceedings of the 2018 International Conference on
  Management of Data, {SIGMOD} Conference 2018, Houston, TX, USA, June 10-15,
  2018}, 2018, pp. 411--426.

\bibitem{D-IDS}
C.~Wickramaarachchi, R.~Kannan, C.~Chelmis, and V.~K. Prasanna, ``Distributed
  exact subgraph matching in small diameter dynamic graphs,'' in \emph{2016
  {IEEE} International Conference on Big Data, BigData 2016, Washington DC,
  USA, December 5-8, 2016}, 2016, pp. 3360--3369.

\bibitem{SymmetryBreaking}
J.~A. Grochow and M.~Kellis, ``Network motif discovery using subgraph
  enumeration and symmetry-breaking,'' in \emph{Research in Computational
  Molecular Biology, 11th Annual International Conference, {RECOMB} 2007,
  Oakland, CA, USA, April 21-25, 2007, Proceedings}, 2007, pp. 92--106.

\bibitem{BacktrackingFramework}
J.~Lee, W.~Han, R.~Kasperovics, and J.~Lee, ``An in-depth comparison of
  subgraph isomorphism algorithms in graph databases,'' \emph{Proc. {VLDB}
  Endow.}, vol.~6, no.~2, pp. 133--144, 2012.

\bibitem{BoostISO}
X.~Ren and J.~Wang, ``Exploiting vertex relationships in speeding up subgraph
  isomorphism over large graphs,'' \emph{Proc. {VLDB} Endow.}, vol.~8, no.~5,
  pp. 617--628, 2015.

\bibitem{Timing}
Y.~Li, L.~Zou, M.~T. {\"{O}}zsu, and D.~Zhao, ``Time constrained continuous
  subgraph search over streaming graphs,'' in \emph{35th {IEEE} International
  Conference on Data Engineering, {ICDE} 2019, Macao, China, April 8-11, 2019},
  2019, pp. 1082--1093.

\bibitem{mislove-2009-socialnetworksthesis}
\BIBentryALTinterwordspacing
A.~Mislove, ``Online social networks: Measurement, analysis, and applications
  to distributed information systems,'' Ph.D. dissertation, Rice University,
  Department of Computer Science, May 2009. [Online]. Available:
  \url{http://socialnetworks.mpi-sws.org/data-wosn2008.html}
\BIBentrySTDinterwordspacing

\bibitem{LDBC-SNB-Data-Generator}
{Linked Data Benchmark Council (LDBC)}, ``Ldbc{-}snb data generator,''
  \url{https://github.com/ldbc/ldbc_snb_datagen}.

\bibitem{LDBC_Graphalytics}
A.~Iosup, T.~Hegeman, W.~L. Ngai, S.~Heldens, A.~Prat{-}P{\'{e}}rez,
  T.~Manhardt, H.~Chafi, M.~Capota, N.~Sundaram, M.~J. Anderson, I.~G. Tanase,
  Y.~Xia, L.~Nai, and P.~A. Boncz, ``{LDBC} graphalytics: {A} benchmark for
  large-scale graph analysis on parallel and distributed platforms,''
  \emph{Proc. {VLDB} Endow.}, vol.~9, no.~13, pp. 1317--1328, 2016.

\bibitem{GraphQL}
H.~He and A.~K. Singh, ``Graphs-at-a-time: query language and access methods
  for graph databases,'' in \emph{Proceedings of the {ACM} {SIGMOD}
  International Conference on Management of Data, {SIGMOD} 2008, Vancouver, BC,
  Canada, June 10-12, 2008}, 2008, pp. 405--418.

\bibitem{SPath}
P.~Zhao and J.~Han, ``On graph query optimization in large networks,''
  \emph{Proc. {VLDB} Endow.}, vol.~3, no.~1, pp. 340--351, 2010.

\bibitem{DAF}
M.~Han, H.~Kim, G.~Gu, K.~Park, and W.~Han, ``Efficient subgraph matching:
  Harmonizing dynamic programming, adaptive matching order, and failing set
  together,'' in \emph{Proceedings of the 2019 International Conference on
  Management of Data, {SIGMOD} Conference 2019, Amsterdam, The Netherlands,
  June 30 - July 5, 2019}, 2019, pp. 1429--1446.

\bibitem{CECI}
B.~Bhattarai, H.~Liu, and H.~H. Huang, ``{CECI:} compact embedding cluster
  index for scalable subgraph matching,'' in \emph{Proceedings of the 2019
  International Conference on Management of Data, {SIGMOD} Conference 2019,
  Amsterdam, The Netherlands, June 30 - July 5, 2019}, 2019, pp. 1447--1462.

\bibitem{RADS}
X.~Ren, J.~Wang, W.~Han, and J.~X. Yu, ``Fast and robust distributed subgraph
  enumeration,'' \emph{Proc. {VLDB} Endow.}, vol.~12, no.~11, pp. 1344--1356,
  2019.

\bibitem{STWig}
Z.~Sun, H.~Wang, H.~Wang, B.~Shao, and J.~Li, ``Efficient subgraph matching on
  billion node graphs,'' \emph{Proc. {VLDB} Endow.}, vol.~5, no.~9, pp.
  788--799, 2012.

\bibitem{SurveyOnTimely}
L.~Lai, Z.~Qing, Z.~Yang, X.~Jin, Z.~Lai, R.~Wang, K.~Hao, X.~Lin, L.~Qin,
  W.~Zhang, Y.~Zhang, Z.~Qian, and J.~Zhou, ``Distributed subgraph matching on
  timely dataflow,'' \emph{Proc. {VLDB} Endow.}, vol.~12, no.~10, pp.
  1099--1112, 2019.

\bibitem{IncIsoMat}
W.~Fan, J.~Li, J.~Luo, Z.~Tan, X.~Wang, and Y.~Wu, ``Incremental graph pattern
  matching,'' in \emph{Proceedings of the {ACM} {SIGMOD} International
  Conference on Management of Data, {SIGMOD} 2011, Athens, Greece, June 12-16,
  2011}, 2011, pp. 925--936.

\bibitem{IncIsoMatJournal}
W.~Fan, X.~Wang, and Y.~Wu, ``Incremental graph pattern matching,'' \emph{{ACM}
  Trans. Database Syst.}, vol.~38, no.~3, pp. 18:1--18:47, 2013.

\bibitem{GraphFlow}
C.~Kankanamge, S.~Sahu, A.~Mhedbhi, J.~Chen, and S.~Salihoglu, ``Graphflow: An
  active graph database,'' in \emph{Proceedings of the 2017 {ACM} International
  Conference on Management of Data, {SIGMOD} Conference 2017, Chicago, IL, USA,
  May 14-19, 2017}, 2017, pp. 1695--1698.

\bibitem{DeltaRules}
A.~Gupta, I.~S. Mumick, and V.~S. Subrahmanian, ``Maintaining views
  incrementally,'' in \emph{Proceedings of the 1993 {ACM} {SIGMOD}
  International Conference on Management of Data, Washington, DC, USA, May
  26-28, 1993}, 1993, pp. 157--166.

\bibitem{GenericJoin}
H.~Q. Ngo, C.~R{\'{e}}, and A.~Rudra, ``Skew strikes back: new developments in
  the theory of join algorithms,'' \emph{{SIGMOD} Rec.}, vol.~42, no.~4, pp.
  5--16, 2013.

\bibitem{SSDJournal}
J.~Gao, C.~Zhou, and J.~X. Yu, ``Toward continuous pattern detection over
  evolving large graph with snapshot isolation,'' \emph{{VLDB} J.}, vol.~25,
  no.~2, pp. 269--290, 2016.

\end{thebibliography}

%

\begin{IEEEbiographynophoto}{Zhaokang Wang} received the BS degree in Nanjing University, China, in 2013.
  He is currently working towards the Ph.D. degree in Nanjing University.
  His research interests include distributed graph algorithms and distributed graph processing systems.
\end{IEEEbiographynophoto}

\begin{IEEEbiographynophoto}{Weiwei Hu} received the BS degree in Hunan University, China, in 2017.
  She is currently working towards the Master degree in Nanjing University.
  Her research interests include distributed subgraph matching.
\end{IEEEbiographynophoto}


\begin{IEEEbiographynophoto}{Chunfeng Yuan}
  is a professor in the computer science department and State Key Laboratory for Novel Software Technology, Nanjing University, China. Her main research interests include computer architecture, parallel and distributed computing.
\end{IEEEbiographynophoto}

\begin{IEEEbiographynophoto}{Rong Gu}
  is an associate research professor at State Key Laboratory for Novel Software Technology, Nanjing University, China. Dr. Gu received the Ph.D. degree in computer science from Nanjing University in December 2016. His research interests include parallel computing, distributed systems and distributed machine learning.
\end{IEEEbiographynophoto}

\begin{IEEEbiographynophoto}{Yihua Huang}
  is a professor in State Key Laboratory for Novel Software Technology, Nanjing University, China. He received his Ph.D. degree in computer science from Nanjing University. His main research interests include parallel and distributed computing, big data parallel processing, and distributed machine learning.
\end{IEEEbiographynophoto}


\vfill

\enlargethispage{-5in}

\end{document}